\documentclass[12pt]{article}
\usepackage{graphicx}
\usepackage{epstopdf}% To incorporate .eps illustrations using PDFLaTeX, etc.
\usepackage{subfigure}% Support for small, `sub' figures and tables
\usepackage{amsfonts,amsthm,amsmath,url}

\theoremstyle{plain}% Theorem-like structures
\newtheorem{theorem}{Theorem}[section]
\newtheorem{corollary}[theorem]{Corollary}
\newtheorem{lemma}[theorem]{Lemma}
\newtheorem{conjecture}[theorem]{Conjecture}

\theoremstyle{definition}

\theoremstyle{remark}
\newtheorem{remark}{Remark}[section]

\def\p{Pain\-lev\'e}

\def\peqs{\p\ equations}
\def\th{\theta}%{\vartheta}

\newcommand{\dPI}{$\textup{dP}_{\small\textup{I}}$}
\newcommand{\PII}{$\textup{P}_{\small\textup{II}}$}
\newcommand{\PIV}{$\textup{P}_{\small\textup{IV}}$}

\numberwithin{equation}{section}
\numberwithin{figure}{section}
\def\arcsec{\mathop{\rm arcsec}\nolimits}
\def\Ai{\mathop{\rm Ai}\nolimits}
\def\Bi{\mathop{\rm Bi}\nolimits}
\def\T{\mbox{\rm T}}
\def\O{\mathcal{O}}
\newcommand{\comment}[1]{}

\def\e{{\rm e}}\def\i{{\rm i}}\def\d{{\rm d}}
\newcommand{\deriv}[3][]{\frac{\d^{#1}{#2}}{{\d{#3}}^{#1}}}

\begin{document}
\title{Unique positive solution for an alternative\\ discrete Painlev\'e I equation}%\dPI}
\author{Peter A.~Clarkson\\
School of Mathematics, Statistics \& Actuarial Science\\ University of Kent, Canterbury, CT2 7NF, UK\\
{Email: P.A.Clarkson@kent.ac.uk} \\[5pt]
Ana F.~Loureiro\\
School of Mathematics, Statistics \& Actuarial Science\\ University of Kent, Canterbury, CT2 7NF, UK\\
{Email: A.Loureiro@kent.ac.uk}\\[5pt]
and\\[5pt]
Walter Van Assche\\
Department of Mathematics, KU Leuven\\ Celestijnenlaan 200 B, Box 2400, BE-3001, Leuven, Belgium\\
{Email: Walter.VanAssche@wis.kuleuven.be}}
\maketitle

\begin{abstract}
We show that the alternative discrete Painlev\'e I equation (alt-\dPI) has a unique solution which remains positive for all $n \geq 0$.
Furthermore, we identify this positive solution in terms of a special solution of the second Painlev\'e equation (\PII) involving the Airy function $\Ai(t)$. The special-function solutions of \PII\ involving only the Airy function $\Ai(t)$ therefore have the property that they remain positive for all $n\geq 0$ and all $t \geq 0$, which is a new characterization of these special solutions of \PII\ and alt-\dPI.
\end{abstract}

\section{\label{Sec1}Introduction}
We will investigate the system of two nonlinear equations
\begin{subequations}\label{dPIsys}\begin{align} 
& a_n+a_{n+1} = b_n^2 -t , \label{dPI-1} \\
 & a_n(b_n+b_{n-1}) = n, \label{dPI-2}
\end{align}\end{subequations}
which is known as the \textit{alternative discrete Painlev\'e I equation}\ (alt-\dPI) \cite{refFokasGR}. These equations arise, for example, when one wants to find the recurrence coefficients of orthogonal polynomials with an exponential cubic weight; see, e.g.~\cite{refBD,refMagnus}.
They also arise when one wants to compute the recurrence coefficients of multiple orthogonal polynomials with an exponential cubic
weight \cite{refFilWVALun}. The equations also give relations between solutions of the equation
\begin{equation} \label{PIIsc}
\deriv[2]{y}{t} = 2y^3 - 2ty - 2\alpha,
\end{equation}
for different values of the parameter $\alpha$ \cite{Clarkson,refGR}. Equation \eqref{PIIsc} is equivalent to the second Painlev\'e equation (\PII)
\begin{equation} \label{eqPII}
\deriv[2]{w}{z} = 2w^3 + zw + \alpha,
\end{equation}
through the scaling
%\begin{equation} \label{dPI-sc} 
$w(z)=-2^{-1/3}\,y(t)$, with $z=-2^{1/3}\,t$.
%\end{equation}
We note that if we solve \eqref{dPI-2} for $a_n$ and substitute this in \eqref{dPI-1}, then we obtain
\begin{equation} \label{dPI-b}
 \frac{n+1}{b_n+b_{n+1}} + \frac{n}{b_n+b_{n-1}} = b_n^2-t.
\end{equation}

The main result in this paper is the following.

\begin{theorem} \label{thmab}
For nonnegative values of $t$, there exists a unique solution of \eqref{dPIsys} with $a_0(t)=0$ for which $a_{n+1}(t) >0$ and $b_n(t) >0$ for all $n \geq 0$, corresponding to the initial value 
\begin{equation} b_0 (t)= -{\Ai'(t)}/{\Ai(t)}, \label{dPI-ic}
%3^{1/3}\,\frac{\Gamma(\tfrac23)}{\Gamma(\tfrac13)}.
\end{equation} where $\Ai(t)$ is the Airy function.
\end{theorem}

The corresponding result for the nonlinear recurrence \eqref{dPI-b} is
\begin{theorem} \label{thmb}
For nonnegative values of $t$, there exists a unique solution of \eqref{dPI-b} for which $b_n(t) > 0$ for all $n \geq 0$, corresponding to the initial value \eqref{dPI-ic}. %$$b_0 (t)= -\frac{\Ai'(t)}{\Ai(t)}. $$
%3^{1/3}\frac{\Gamma(\tfrac23)}{\Gamma(\tfrac13)}.
\end{theorem}

We observe that the value of $b_1(t)$ is fixed once $a_0(t)$ and $b_0(t)$ are known. Specifically one has $b_1(t)=-b_0(t)+{1}/{[b^2_0(t)-t]}$. Hence the solutions of \eqref{dPI-b} only depend on \textit{one} free parameter $b_0(t)$ even though it is a recurrence relation of second order. This is reflected in Theorem \ref{thmab} by the initial value $a_0(t)=0$. It is important to note that the initial value \eqref{dPI-ic} involves \textit{only} $\Ai(t)$, rather than a linear combination of the Airy functions $\Ai(t)$ and $\Bi(t)$.

Uniqueness of positive solutions for non-linear recurrence relations was already studied earlier for the equation
\begin{equation} \label{dPI}
 x_n(x_{n+1}+x_n+x_{n-1}) = n, 
\end{equation}
which is known as a discrete Painlev\'e I equation \cite{refFIKa,refMagnus95} and which appears naturally when one wants to find the recurrence coefficients of the orthogonal polynomials with the exponential weight $\exp(-x^4)$ on the real line, the so-called \textit{Freud weight}. Lew and Quarles \cite{refLQ},  Nevai \cite{refNevai}, and Bonan and Nevai \cite{refBonanNevai} showed that there exists a unique solution of \eqref{dPI} with $x_0=0$ for 
which $x_n >0$ for all $n \geq 1$ and this solution corresponds to $x_1 =2\Gamma(\tfrac34)/\Gamma(\tfrac14)$. This result was recently also proved for a family of non-linear recurrence relations generalizing \eqref{dPI}, see \cite{refANSV}.

The fact that there is a unique solution of \eqref{dPI}, or equivalently for \eqref{dPI-b}, means that it is not wise to compute this positive solution starting from the initial values $a_0(t)=0$ and $b_0(t)= -{\Ai'(t)}/{\Ai(t)}%\sqrt[3]{3}\Gamma(\tfrac23)/\Gamma(\tfrac13)
$,
because a small error in $b_0$ means that one will not be generating the positive solution and after some time there will be negative $b_n(t)$. In \S\ref{sec2} we will give a method to compute the positive solution $(b_n)_{n \geq 0}$ %where $b_n:=b_n(t)$ 
by means of a
fixed point algorithm using a contraction acting on infinite sequences. That method is stable and is very suitable to compute the
positive solution numerically. %Further in \S\ref{Sec5}, we give an explicit expression for $b_n$ in terms of the Airy function $\Ai(t)$.

One can easily obtain the asymptotic behaviour of the positive solution of \eqref{dPI}. This was already proved in \cite{refFilWVALun}, 
but we repeat the proof here for completeness.
\begin{corollary}
For the positive solution of \eqref{dPI} one has
$\lim_{n \to \infty} {a_n}/{n^{2/3}} = \tfrac12$ and\break $\lim_{n \to \infty} {b_n}/{n^{1/3}} = 1$.
\end{corollary}
\begin{proof}
The positivity of the $a_n$ and \eqref{dPI-1} imply that $a_n \leq b_n^2$. The positivity of the $b_n$ and \eqref{dPI-2}
imply $a_nb_n \leq n$. Together this implies $a_n \leq n^2/a_n^2$ so that $a_n \leq n^{2/3}$. From \eqref{dPI-1} it then follows that
$b_n^2 \leq n^{2/3} + (n+1)^{2/3}$. We can then define 
\[ B_1 = \liminf_{n \to \infty} {b_n}/{n^{1/3}}, \qquad B_2 = \limsup_{n \to \infty} {b_n}/{n^{1/3}}. \]
If we take a subsequence such that $b_n/n^{1/3} \to B_1$, then from \eqref{dPI-b} one finds
\[ B_1^2 \geq {2}/{(B_1+B_2)}. \]
In a similar way, by taking a subsequence such that $b_n/n^{1/3}\to B_2$, it follows from \eqref{dPI-b} that
\[ B_2^2 \leq {2}/{(B_1+B_2)}. \]
Together this implies that $B_2^2 \leq B_1^2$, and since one always has $B_1 \leq B_2$, it follows that $B_1=B_2$ so that $b_n/n^{1/3}$ converges.
From \eqref{dPI-2} it then also follows that $a_n/n^{2/3}$ converges. Taking limits in \eqref{dPI-1} and \eqref{dPI-2} then easily gives the values
of these limits.
\end{proof}

This paper is organised as follows. In \S\ref{Sec2} we discuss the uniqueness of $b_n$ and discuss its behaviour in \S\ref{Sec3}. In \S\ref{Sec4} we discuss the relationship with solutions of \PII\ \eqref{eqPII}.
It is well known that \PII\ \eqref{eqPII} has special-function solutions expressible in terms Airy functions \cite{Clarkson,ForrW01,refOkamotoPIIPIV}. Using these we give explicit expressions for the solutions of \eqref{dPIsys} and \eqref{dPI-b} in terms of the Airy function $\Ai(t)$ in \S\ref{Sec5}.
The special-function solutions of \PII\ involving \textit{only} the Airy function $\Ai(t)$ therefore have the property that they remain positive for all $n\geq 0$ and all $t \geq 0$, which is a new characterization of these special solutions of \PII\ and alt-\dPI.

\section{\label{Sec2}Proof of the uniqueness} \label{sec2}
In this section we will prove Theorem \ref{thmb}. A positive solution $(b_n)_{n \geq 0}$ implies that $a_0=0$ and $a_n >0$ for $n \geq 1$ by using
\eqref{dPI-2}, so that Theorem \ref{thmab} follows as well. The idea of the proof is to construct a mapping $\T$ on the space $\mathbb{R}_+^{\mathbb{N}}$
of positive sequences $(x_n)_{n \geq 0}$ and to show that this is a contraction on a complete set in this space. The unique fixed point 
will be the desired positive solution of \eqref{dPI-b}.

Let $(x_n)_{n \geq 0}$ be an infinite sequence with $x_n > 0$ for all $n \geq 0$. We define a new sequence $(({\T}x)_n)_{n > 0}$
implicitly by
\begin{equation} \label{Timp}
 \frac{n+1}{({\T}x)_n +x_{n+1}} + \frac{n}{({\T}x)_n+x_{n-1}} = ({\T}x)_n^2 - t , \qquad n \geq 1 
\end{equation}
and
\begin{equation} \label{T0imp}
 \frac{1}{({\T}x)_0+x_1} = [({\T}x)_0]^2-t. 
\end{equation}
Observe that $({\T}x)_n$ is a solution of
\begin{equation} \label{eqy}
 \frac{n+1}{y+x_{n+1}} + \frac{n}{y+x_{n-1}} = y^2-t . 
\end{equation}
If $x_{n+1} > 0$ and $x_{n-1} > 0$, then the left hand side of \eqref{eqy} is a positive and decreasing function of $y \in (0,\infty)$ and
the right hand side of \eqref{eqy} is obviously increasing on $[0,\infty)$ so that \eqref{eqy} has a unique positive solution
and therefore for $n \geq 0$ we set
%\[ ({\T}x)_n = \textrm{positive solution of \eqref{eqy}}. \]
$({\T}x)_n$ to be the positive solution of \eqref{eqy}. 
Note that the left hand side requires $x_{n+1}$ and $x_{n-1}$ for $n \geq 1$ but only requires $x_1$ for $n=0$.
Equation \eqref{eqy} corresponds to a quartic equation in $y$ for $n \geq 1$ and a cubic equation if $n=0$. There is a negative solution
between $-x_{n+1}$ and $-x_{n-1}$ when $n \geq 1$ and two complex conjugate solutions or two real negative solutions which will not be used in
the remainder of this paper.

We will now give some properties of this mapping ${\T}$.
\begin{lemma} \label{lem21}
If $x_n > 0$ for all $n \geq 0$, then $0< ({\T}x)_n \leq B_n(t)$ with 
\begin{align}
	B_n(t) = (n+\tfrac{1}{2})^{1/3}\Bigg\{
				&\left[ 1 - \sqrt{1 - \frac{t ^3}{27(n+\tfrac{1}{2})^2 }}\;\right]^{1/3}%\nonumber\\ &\qquad 
				+\left[ 1 + \sqrt{1 - \frac{t ^3}{27(n+\tfrac{1}{2})^2 }}\;\right]^{1/3} \Bigg\}.
\label{Rnt}\end{align}
%for all $n \geq 0$. 
\end{lemma}
\begin{proof}
%Since $x_{n+1} \geq 0$ and $x_{n-1} \geq 0$, it follows from \eqref{Timp} that
%\[ ({\T}x)_n^2 \leq \frac{n+1}{({\T}x)_n} +\frac{n}{({\T}x)_n} = \frac{2n+1}{({\T}x)_n}, \]
%so that $({\T}x)_n^3 \leq 2n+1$. Together with the positivity of $({\T}x)_n$, this gives the required bounds.
Since $x_{n+1}>0$ and $x_{n-1}>0$, it follows from \eqref{Timp}
$$
	x_n = ({\T}x)_n = \sqrt{t+\frac{n+1}{({\T}x)_n + x_{n+1}} + \frac{n}{({\T}x)_n + x_{n-1}}} 
	\leq \sqrt{t+\frac{2n+1}{({\T}x)_n } } 
$$
and therefore 
$$
	 ({\T}x)_n \Big\{\big[ ({\T}x)_n \big]^2 - t\Big\}\leq 2n+1. 
$$
The roots of the equation 
$y(y^2 -t ) = 2n+1$
can be written as: 
{\small 
\begin{align*}
	 y_{1} %&= \frac{1}{2^{1/3}} \left( \frac{2^{2/3}}{3 \left( 2n+1 + \sqrt{(2n+1)^2 - \frac{4}{27} t^3 }\right)^{1/3}}
		%	+ \left( 2n+1 + \sqrt{(2n+1)^2 - \frac{4}{27} t^3 }\right)^{1/3}
		%	\right) \\
		&=%\frac{(2n+1)^{1/3}}{2^{1/3}} 
	(n+\tfrac{1}{2})^{1/3}\left\{ %\left\{\left[ 2n+1 - \sqrt{(2n+1)^2 - \tfrac{4}{27} t^3 }\;\right]^{1/3}
		\left[ 1 - \sqrt{ 1 - \frac{t^3}{27(n+\tfrac12)^2} }\;\right]^{1/3} +\left[ 1 + \sqrt{ 1 - \frac{t^3}{27(n+\tfrac12)^2} }\;\right]^{1/3}\right\}\\
		%+ \left[ 2n+1 + \sqrt{(2n+1)^2 - \tfrac{4}{27} t^3 }\;\right]^{1/3}\right\} \\ 
	y_{2^\pm} %&= - \frac{(1+\i \sqrt{3})}{ 2^{2/3} 3 \left[ 2n+1 + \sqrt{(2n+1)^2 - \frac{4}{27} t^3 }\;\right]^{1/3}}
		%	- \frac{(1- i \sqrt{3}) \left[ 2n+1 + \sqrt{(2n+1)^2 - \frac{4}{27} t^3 }\;\right]^{1/3}}{2^{4/3}}\\
		 & = - \frac{(n+\tfrac12)^{1/3}}{ 2^{1/3} 3 } \Bigg\{ (1\pm\i \sqrt{3})\left[ 1 - \sqrt{ 1 - \frac{t^3}{27(n+\tfrac12)^2} }\;\right]^{1/3} 
		 \\ &\qquad\qquad\qquad\qquad\quad 
		 + (1\mp\i \sqrt{3})\left[ 1 + \sqrt{1 -\frac{t^3}{27(n+\tfrac12)^2} }\;\right]^{1/3} \Bigg\}\\
\comment{	y_{1^-} %&= - \frac{(1+\i \sqrt{3})}{ 2^{2/3} 3 \left[ 2n+1 + \sqrt{(2n+1)^2 - \frac{4}{27} t^3 }\;\right]^{1/3}}
		%	- \frac{(1- i \sqrt{3}) \left[ 2n+1 + \sqrt{(2n+1)^2 - \frac{4}{27} t^3 }\;\right]^{1/3}}{2^{4/3}}\\
		 & = - \frac{(2n+1)^{1/3}}{ 2^{2/3} 3 } \Bigg\{ (1+\i \sqrt{3})\left[ 1 - \sqrt{ 1 - \frac{t^3}{27(n+\tfrac12)^2} }\;\right]^{1/3} \\
		 &\qquad\qquad\qquad\qquad\quad + (1-\i \sqrt{3})\left[ 1 + \sqrt{1 -\frac{t^3}{27(n+\tfrac12)^2} }\;\right]^{1/3} \Bigg\}\\
	y_{2^-} %&= - \frac{(1+\i \sqrt{3})}{ 2^{2/3} 3 \left[ 2n+1 + \sqrt{(2n+1)^2 - \frac{4}{27} t^3 }\;\right]^{1/3}}
		%	- \frac{(1- i \sqrt{3}) \left[ 2n+1 + \sqrt{(2n+1)^2 - \frac{4}{27} t^3 }\;\right]^{1/3}}{2^{4/3}}\\
		 & = - \frac{(2n+1)^{1/3}}{ 2^{2/3} 3 } \Bigg\{ (1-\i \sqrt{3})\left[ 1 - \sqrt{1 - \frac{t^3}{27(n+\tfrac12)^2} }\;\right]^{1/3} \\
		 &\qquad\qquad\qquad\qquad\quad+ (1+\i \sqrt{3})\left[ 1 + \sqrt{1 -\frac{t^3}{27(n+\tfrac12)^2} }\;\right]^{1/3} \Bigg\}}
\end{align*}}%
In fact $y_{1} = (n+\tfrac{1}{2})^{1/3} Q(w)$, with $w=\tfrac13{t}(n+\tfrac12)^{-2/3}$, where 
\begin{equation}\label{Qfunction}
        Q(w) = \left( 1 - \sqrt{1 - w^3}\right)^{1/3}
                                +\left( 1 + \sqrt{1-w^3}\right)^{1/3}. 
\end{equation}
Notice that for $w>1$ it follows that $Q(w) = 2 \sqrt{w} \cos\big[\tfrac{1}{3} \arccos (w^{-3/2})\big]$. Hence, $Q(w)$ is always positive for any real value of $w$, in fact it is the unique positive solution of the equation 
\begin{equation}\label{Qfunctioneq}
        Q(w) \left[Q^2(w) - 3w  \right] =2. 
\end{equation}
As a consequence, $y_{1}$ is positive for any real value of $t$ (and any integer $n$).
%The root $y_{1}$ is positive for any real value of $t$. 
The roots $y_{2^\pm} $ are either a pair of complex conjugate numbers, when $(2n+1)^2 > \frac{4}{27} t^3$, or two negative numbers, when $(2n+1)^2 < \frac{4}{27} t^3 $. The case where eventually $(2n+1)^2 =\frac{4}{27} t^3 $ would produce a positive single root $y_{1}$ and a double negative root $y_{1^-}$. Hence, no matter the value of $t$, there will be a unique positive root, which is $y_{1}=B_n(t)$. 
Now the positivity of $({\T}x)_n $ yields the required bounds. 
\end{proof}

Let us write 
$B_n(t) = R\left( n+\tfrac{1}{2}, t\right) $
with 
$R(z,t) = z^{1/3} Q\left({t}/({3z^{2/3}})\right)$,
where $Q(w)$ is the function given by \eqref{Qfunction}. 

\begin{lemma} \label{lem Rzt} The function 
$$
	R(z, t) = z^{1/3}\left\{ \left[ 1 - \sqrt{1 - \frac{t^3}{27z^{2}}}\;\right]^{1/3}
				+\left[ 1 + \sqrt{1 - \frac{t^3}{27z^{2}}}\;\right]^{1/3} \right\} 
$$
is increasing and concave in $z$ on $(0,\infty)$ and increasing in $t$ on $\mathbb{R}$. 
\end{lemma}
\begin{proof} 
Since the function $Q(w)$ is the unique positive root of the equation \eqref{Qfunctioneq} it follows that $Q^2(w) > 3w $ and therefore  $Q^2(w) > w $ for any $w\geq 0$. Now, by differentiating \eqref{Qfunctioneq} once we obtain 
$$
        Q'(w) \left[Q^2(w) - w\right] = Q(w) >0,
$$ 
and consequently $Q'(w) >0$ and also  
$$ 
        Q(w) - 2 w Q'(w) =\frac{Q(w)}{Q^2(w) - w} \left[ Q^2(w) - 3 w \right] >0.
$$ 
A second differentiation of \eqref{Qfunctioneq} gives 
$$
        Q''(w) \left[Q^2(w) - w  \right] = 2 \big[1 - Q(w)Q'(w)\big] Q'(w),
$$ 
which can be written as 
$$
        Q''(w) \left[Q^2(w) - w  \right]^3 = - 2 w Q(w),
$$ 
and therefore we conclude that $Q''(w) < 0 $ for  any $w> 0$, and $Q''(w) > 0 $ if $w<0$. 

Now, from the definition of the function $R(z,t)$, observe that   
\begin{align*}
        \frac{\partial R}{\partial z}    &= \frac{1}{3 z^{2/3}} \left. \Big[ Q(w) - 2 w Q'(w) \Big]\right|_{w={t}/{(3z^{2/3})}} %\\&
                                                  =   \frac{1}{3 z^{2/3}} \left.  \frac{Q(w)  \left[ Q^2(w) - 3 w \right] }{Q^2(w) - w}\right|_{w={t}/{(3z^{2/3})}},
\end{align*}
and also that 
\begin{align*}     
        \frac{\partial ^2R}{\partial z^2} &= \frac{2}{9 z^{5/3}} \left. \Big[- Q(w) +3 w Q'(w) + 2 w^2 Q''(w)\Big]\right|_{w={t}/{(3z^{2/3})}}\\
                        &= \left\{ \begin{array}{lcl} 
                                \left.  \dfrac{2Q''(w)}{9 w z^{5/3}}\right|_{w={t}/{(3z^{2/3})}}, &\quad \text{if} & t>0, \\
                                - {2^{4/3}}/{(9z^{5/3})}, &\quad \text{if} & t=0.
                        \end{array}\right. 
                        %= \frac{2Q''(\tfrac{t}{3z^{2/3}})}{3 t z}. 
\end{align*}
Consequently, 
$\displaystyle\frac{\partial R}{\partial z} > 0$,  $\displaystyle\frac{\partial ^2R}{\partial z^2} < 0$ and
$\displaystyle\frac{\partial R}{\partial t}  > 0$.
\end{proof}

\begin{lemma} \label{lem22}
If $x_n \leq B_n(t)$ for all $n \geq 0$, then $({\T}x)_n \geq c_1 B_n(t)$ for all $n \geq 1$ and $({\T}x)_0 \geq c_0$, where
$c_0=0.68554389$ and $c_1=0.6379714$.
\end{lemma}
\begin{proof}
Equation \eqref{Timp} easily gives
\begin{equation} \label{Tlm2}
 ({\T}x)_n^2 - t \geq \frac{n+1}{({\T}x)_n+ B_{n+1}(t)} + \frac{n}{({\T}x)_n+ B_{n-1}(t) }
	 \ \ \text{for} \ \ n\geq 0. 
\end{equation} 
With Lemma \ref{lem Rzt}, we readily observe that for $z \in [0,\infty)$ the function $f(z) = {1}/{[y+R(z,t)]}$ is such that 
$$\frac{\partial f }{\partial z} = -\frac{1}{[y+R(z,t)]^2} \frac{\partial R}{\partial z}\leq 0,$$
whilst
$$\frac{\partial^2 f}{\partial z^2} = \frac{1}{[y+R(z,t)]^3}\left(\frac{\partial R}{\partial z}\right)^2
			 - \frac{1}{[y+R(z,t)]^2}\frac{\partial^2 R}{\partial z^2} \geq 0. 
$$
Hence $f(z)$ is convex on $(0,\infty)$, so that 
\[ \lambda f(z_1) + (1-\lambda) f(z_2) \geq f(\lambda z_1 + (1-\lambda)z_2) , \qquad \lambda \in [0,1], \]
whenever $0 < z_1 \leq z_2 < \infty$. For $n \geq 1$ we choose $z_1=n+\tfrac32$, $z_2=n-\tfrac12$ and $\lambda=\tfrac12$, then this gives for $y=({\T}x)_n$
$$
	\frac{1}{({\T}x)_n+ R\left( n+ \frac{3}{2} , t \right) } + \frac{1}{({\T}x)_n+ R\left( n- \tfrac12 , t \right) }
	\geq \frac{2}{({\T}x)_n+ R\left( n+ \tfrac12, t \right)}
	 \ \ \text{for} \ \ n\geq 1,
$$
which easily leads to 
$$
	\frac{n+1}{({\T}x)_n+ R\left( n+ \frac{3}{2} , t \right) } + \frac{n}{({\T}x)_n+ R\left( n- \tfrac12 , t \right) }
	\geq \frac{2n}{({\T}x)_n+ R\left( n+ \tfrac12, t \right)}
	 \ \ \text{for} \ \ n\geq 1. 
$$
Combined with \eqref{Tlm2}, and also because $B_n(t) = R(n+\tfrac12,t))$, this gives
$$({\T}x)_n^2 - t 		\geq \frac{2n}{({\T}x)_n+ B_{n}(t)} , 
	 \ \ \text{for} \ \ n\geq 1.$$
%\[ ({\T}x)_n^2 \geq \frac{2n}{({\T}x)_n+(2n+1)^{1/3}}. \]
and therefore 
$$
	\left[({\T}x)_n^2 - t\right] \left[ ({\T}x)_n+ B_{n}(t)\right] \geq 2n, \qquad \text{for} \quad n\geq 1. 
$$

With $c = {({\T}x)_n}/{ B_{n}(t)}$, the latter inequality %becomes 
$$(c+1) \Big[c^2 \left( \, B_{n}(t)\right)^2 - t \Big] B_{n}(t) \geq 2n, \qquad\text{for} \quad n\geq 1, $$
which can be expressed as 
$$
	(c+1) \Big[(c^2-1) B^3_{n}(t) + (2n+1) \Big] \geq 2n,\qquad\text{for} \quad n\geq 1, 
$$
because 
$t B_{n}(t) = B^3_{n}(t) -(2n+1)$. %\footnote{We also have $Q(y)^3 - 3 y Q(y)=2$.} 
%Now, considering that 
Since $	B_{n}(t) = \left(n+\tfrac12\right)^{1/3}Q(w_n(t))$, for $n\geq 0$, with $\displaystyle w_n(t) = {t}/{[3(n+\tfrac12)^{2/3}]}$, it follows 
%$$
%	(c+1) \Big((c^2-1) \left(n+\tfrac{1}{2}\right) \Big(Q(w_n(t))\Big)^3+ (2n+1) \Big) \geq 2n, \ \text{ for any } n\geq 1,
%$$
%i.e., 
\begin{equation}\label{ineq n}
	(c+1) \left\{(c^2-1) \Big[Q(w_n(t))\Big]^3+ 2 \right\} \geq \frac{4n}{2n+1}, \quad \text{ for any }\quad n\geq 1. 
\end{equation}
For $t\geq0$, the sequence $\left(w_n(t)={t}/{[3 (n+\tfrac12)^{2/3}]}\right)_{n\geq 1}$, is monotonically decreasing so that $0 < w_n (t) \leq w_1(t) ,\ n\geq 1$. Since $Q(z)$ is a positive, increasing function (with 
$ \lim_{z\to -\infty} Q(z) =0$, $\lim_{z\to +\infty} Q(z) =+\infty$), 
it follows that $$ Q(0) \leq Q(w_n (t)) \leq Q(w_1(t)) ,$$ for any $t\geq 0$ and $ n\geq 1 $, 
where $Q(0) = 2^{1/3}$ and
$$Q(w_1(t)) =\left( 1 - \sqrt{1 - w_1^3(t) }\right)^{1/3}+\left( 1 + \sqrt{1-w_1^3(t)}\right)^{1/3} .$$
Since $0< ({\T}x)_n \leq B_{n}(t)$ implies $0<c<1$ and therefore $(c+1)(c^2-1)<0$, then 
$$2c^2 (c+1) \ \geq \ (c+1) \left\{(c^2-1) \left[Q(w_n (t))\right]^3 +2\right\} \geq 
%	\ \geq (c+1) \Big( (c^2-1) \Big(Q(y_1(t))\Big)^3 + 2 \Big) \ 
\frac{4n}{2n+1} \geq \frac{4}{3}.$$
Hence, $c$ is such that $c^2 (c+1) \geq \tfrac{2}{3}$. 
The function $2c^2 (c+1) $ is monotone increasing on $[0,1]$, and this implies $c\geq c_1$, where $c_1$ is the positive root of $c^2 (c+1)=\frac{2}{3}$, which is approximately $c_1= 0.6379714$. 

For $n=0$ we use \eqref{Tlm2} to obtain 
$$[({\T}x)_0]^2 - t \geq {1}/{[({\T}x)_0+ B_{1}(t) ]},$$
so that $\left[({\T}x)_0+ B_{1}(t)\right]\left\{ [({\T}x)_0]^2 - t \right\}\geq 1$. With $\displaystyle d= {({\T}x)_0}/{B_{1}(t)}$, then we have 
$$(d+1)\left[ d^2 B^3_{1}(t)- t B_{1}(t)\right]\geq 1. $$
With the identity $t R\left( \tfrac{3}{2},t\right) = B^3_{1}(t) -3$, the latter inequality becomes 
$$
	(d+1)\left[ (d^2-1) B^3_{1}(t)+3\right]\geq 1. 
$$
which, because of the fact that $B_{1}(t) \geq 3^{1/3}$ and $0\leq d \leq 1$, implies 
$$
	(d+1)\left[ (d^2-1)3+3\right]\geq(d+1)\left[ (d^2-1) B^3_{1}(t)+3\right]\geq 1.
$$
The function $d^2(d+1)$ is monotone increasing on $[0,1]$, and this implies $d\geq c_0$, where $c_0$ is the positive root of 
$d^2(d+1)=\frac{1}{3}$, which is approximately $c_0= 0.47533$, so that 
$
({\T}x)_0\geq 0.47533 \ B_{1}(t) \geq 0.47533 \times 3^{1/3} = 0.685544.$
%
%%%%%%%%%%%%%%%%%%%%%%%%%%%%%
%
%
%Let $({\T}x)_n/(2n+1)^{1/3} = y$, then 
%\[ y^2 \geq \frac{2n}{2n+1} \frac{1}{y+1}, \]
%or for $n \geq 1$
%\[ y^2(y+1) \geq \frac{2n}{2n+1} \geq \frac{2}{3}. \]
%Since $y^2(y+1)$ is monotone increasing on $[0,\infty)$, this implies $y \geq c_1$, where $c_1$ is the positive root of $y^3+y^2=\tfrac23$.
%This root is approximately $c_1=0.6379714$.
%
% 
%
%For $n=0$ we use \eqref{T0imp} to find
%\[ [({\T}x)_0]^2 \geq \frac{1}{({\T}x)_0+3^{1/3}} \]
%so that $[({\T}x)_0]^2 (({\T}x)_n+3^{1/3}) \geq 1$,
%so that $({\T}x)_0 \geq c_0$, where $c_0$ is the positive root of $y^3+3^{1/3}y^2 = 1$. This root is approximately $c_0=0.68554389$.
\end{proof}

If we use the norm
$ \|x\| = \sup_{n \geq 0} {|x_n|}/{B_n(t)}$,
with $B_n(t)$ as given in \eqref{Rnt}, 
then Lemma \ref{lem21} implies that $\|{\T}x\| \leq 1$ whenever $(x_n)_{n \geq 0}$ is a positive sequence. Lemma \ref{lem22} implies that
$\|\T^2x\| \geq c_1$. We are interested in the iterations $\T^kx$ for $k \to \infty$, so without loss of generality we may assume
that $c_1B_n(t) \leq x_n \leq B_n(t)$. The mapping ${\T}$ is then a mapping from positive sequences in the unit ball 
$\mathbf{B}_1=\{(x_n)_{n\geq 0} :
 \|x\| \leq 1\}$ to positive sequences in the unit ball $\mathbf{B}_1$.

We will show that whenever $t\geq 0$, ${\T}$ is a contraction on the positive sequences for which $c_1 \leq \|x\| \leq 1$. Let $x,y$ be positive sequences
with $c_1 \leq \|x\|, \|y \| \leq 1$, then from \eqref{T0imp}
\begin{align*}
 \big|({\T}x)_0-({\T}y)_0\big| &= \left| \sqrt{\frac{1}{({\T}x)_0+x_1}+t} - \sqrt{\frac{1}{({\T}y)_0+y_1}+t} \right| \\
 &= \frac{1}{({\T}x)_0 + ({\T}y)_0} \left| \frac{1}{({\T}x)_0+x_1} - \frac{1}{({\T}y)_0+y_1} \right|. 
\end{align*}
From Lemma \ref{lem22} it follows that $({\T}x)_0 \geq c_0 B_1(t)$ and $({\T}y)_0 \geq c_0 B_1(t)$, hence
\[ \big|({\T}x)_0-({\T}y)_0\big| \leq \frac{1}{2c_0 B_1(t)} \frac{|({\T}y)_0-({\T}x)_0 + y_1-x_1 |}{[({\T}x)_0+x_1] [({\T}y)_0+y_1]} . \]
Lemma \ref{lem22} also gives $x_1 \geq c_1 B_1(t) $ and $y_1 \geq c_1 B_1(t)$ so that
\[ \big|({\T}x)_0-({\T}y)_0\big| \leq \frac{1}{2c_0(c_0+c_1)^2 B^3_{1}(t)} \Bigl[\big|({\T}x)_0-({\T}y)_0\big| + |x_1-y_1|\Bigr]. \]
 From this we find
\[ B_0(t)\left[ 1 -\frac{1}{2c_0(c_0+c_1)^2 B^3_{1}(t)} \right] \frac{\big|({\T}x)_0-({\T}y)_0\big| }{B_0(t) }
	\leq \frac{1}{2c_0(c_0+c_1)^2 B^2_{1}(t)} \frac{|x_1-y_1|}{B_1(t) } .\]
For $t\geq 0$, the function $\displaystyle\left[ 1 -\frac{1}{2c_0(c_0+c_1)^2 B^3_{1}(t)} \right]>0$, so that we can write 
\[ \frac{\big|({\T}x)_0-({\T}y)_0\big| }{B_0(t) } \leq \frac{B_1(t)}{B_0(t)\left[ 2c_1(c_0+c_1)^2 B^3_{1}(t) - 1 \right] } \frac{|x_1-y_1|}{B_1(t) } \]
The function on the right is a decreasing function on $t\geq 0$, so that 
\[ \frac{B_1(t)}{B_0(t)\left[ 2c_0(c_0+c_1)^2 B^3_{1}(t) - 1 \right] } 
		\geq \frac{B_1(0) }{B_0(0)\left[ 2c_0(c_0+c_1)^2 B^3_{0}(t) - 1 \right] }
		= 0.568967, \]
approximately. For this reason, we have 
\begin{equation} \label{Txy0t}
 \frac{\big|({\T}x)_0-({\T}y)_0\big| }{B_0(t) } \leq c_2 \frac{|x_1-y_1|}{B_1(t) }, \qquad c_2 = 0.568967\ldots .
\end{equation}
A similar calculation holds for $n \geq 1$: from \eqref{Timp} we find
\begin{align*}
 |({\T}x)_n-({\T}y)_n| &= \left| \sqrt{\frac{n+1}{({\T}x)_n+x_{n+1}} + \frac{n}{({\T}x)_n+x_{n-1}}+t} \right.\\
&\qquad\qquad \left. - \sqrt{\frac{n+1}{({\T}y)_n+y_{n+1}} + \frac{n}{({\T}y)_n+y_{n-1}}+t} \right| \\
 &= \frac{1}{({\T}x)_n + ({\T}y)_n} \left\{ \left| \frac{n+1}{({\T}x)_n+x_{n+1}} - \frac{n+1}{({\T}y)_n+y_{n+1}} \right| \right. \\
	& \qquad\qquad \left. + \left| \frac{n}{({\T}x)_n+x_{n-1}} - \frac{n}{({\T}y)_n+y_{n-1}} \right| \right\}. 
\end{align*}
Lemma \ref{lem22} gives $({\T}x)_n \geq c_1B_n(t)$ and $({\T}y)_n \geq c_1B_n(t)$ so that
\begin{align*}
 |({\T}x)_n-({\T}y)_n| &\leq \frac{1}{2c_1B_n(t)} 
 \left\{ \frac{(n+1)\bigl[ |({\T}y)_n-({\T}x)_n|+ |y_{n+1}-x_{n+1}| \bigr]}{[({\T}x)_n+x_{n+1}] [({\T}y)_n+y_{n+1}]} \right. \\
 & \qquad\qquad\qquad\quad \left. + \frac{n\bigl[ |({\T}y)_n-({\T}x)_n| + |y_{n-1}-x_{n-1}| \bigr]}{[({\T}x)_n+x_{n-1}] [({\T}y)_n+y_{n-1}]} \right\}.
\end{align*}
Since $x_{n+1},y_{n+1} \geq c_1 B_{n+1}(t)$ and $x_{n-1},y_{n-1} \geq c_1B_{n-1}(t)$ we have (together with the lower bound from
Lemma \ref{lem22})
\begin{align*}
 {|({\T}x)_n-({\T}y)_n|} %& & \\
 &\leq \frac{|({\T}x)_n-({\T}y)_n|}{2c_1^3B_{n}(t) } \left\{ \frac{n+1}{ \left[B_n(t) +B_{n+1}(t)\right]^2}
 + \frac{n}{\left[B_n(t) +B_{n-1}(t) \right]^2} \right\} \\
 &\qquad + \frac{\|x-y\|}{2c_1^3B_n(t) } \left\{\frac{(n+1)B_{n+1}(t)}{ \left[B_n(t) +B_{n+1}(t)\right]^2}
 + \frac{nB_{n-1}(t)}{\left[B_n(t) +B_{n-1}(t) \right]^2} \right\}, 
\end{align*}
which can be rearranged so that 
\begin{eqnarray*}
 &&{\left\{1- \frac{1}{2c_1^3B_{n}(t) } \left\{ \frac{n+1}{ \left[B_n(t) +B_{n+1}(t)\right]^2}
 + \frac{n}{\left[B_n(t) +B_{n-1}(t) \right]^2} \right\} \right\} \frac{|({\T}x)_n-({\T}y)_n|}{B_n(t) }} \\
 &&\qquad\qquad\leq 
 \quad \frac{\|x-y\|}{2c_1^3B_n(t)^2 } \left\{\frac{(n+1)B_{n+1}(t)}{ \left[B_n(t) +B_{n+1}(t)\right]^2}
 + \frac{nB_{n-1}(t)}{\left[B_n(t) +B_{n-1}(t) \right]^2} \right\}.
\end{eqnarray*}

The function on the left 
$$
	f(n,t) = 1- \frac{1}{2c_1^3B_{n}(t) } \left\{ \frac{n+1}{ \left[B_n(t) +B_{n+1}(t)\right]^2}
 + \frac{n}{\left[B_n(t) +B_{n-1}(t) \right]^2} \right\} 
$$
is bounded from below by $f(n,0)$ since $B_{n}(t)\geq (2n+1)^{1/3}$ for any $n\geq 0$. In addition, $f(n,0)$ is an increasing function of $n$ and therefore $f(n,0)$ is bounded from below by its value at $n=1$, which is approximately $f(1,0)=0.507422$. Since $f(n,t)\geq 0.507422>0$, we can write 
\begin{equation*}
\frac{|({\T}x)_n-({\T}y)_n|}{B_{n}(t)} \leq 
 \frac{\displaystyle%\frac{1}{2c_1^3B_n(t)^2 } \left\{
 \frac{(n+1)B_{n+1}(t)}{ \left[B_n(t) +B_{n+1}(t)\right]^2}
 + \frac{nB_{n-1}(t)^2}{\left[B_{n}(t) +B_{n-1}(t) \right]^2} %\right\}
 }{\displaystyle %1- \frac{1}{
 2c_1^3B_{n}(t) %} 
 -\left\{\frac{n+1}{ \left[B_n(t) +B_{n+1}(t)\right]^2}+ \frac{n}{\left[B_n(t) +B_{n-1}(t) \right]^2} \right\} 
 } \|x-y\|.
\end{equation*}

We observe that for each $t\geq 0$, the expression on the right is an increasing sequence for $n\geq 1$, therefore it is bounded from above by the values as $n \to \infty$, which is approximately $\displaystyle c_3= %=\frac{ \frac{1}{2c_1^3} \frac{1}{4}}{1-\frac{1}{2c_1^3} \frac{1}{4}}
1/(8c_1^3-1)=0.928273$. Consequently, %we obtain 
for any $n\geq 1$
\begin{equation} \label{Txynt}
 \frac{|({\T}x)_n-({\T}y)_n|}{B_{n}(t)} \leq c_3 \|x-y\|, \qquad c_3 = 0.928273 \ldots . 
\end{equation}
Combining \eqref{Txy0t} and \eqref{Txynt} then gives
\[ \|{\T}x-{\T}y\| \leq 0.928273 \|x-y\|, \]
which shows that ${\T}$ is a contraction. Since the unit ball with the norm $\|\cdot\|$ is complete, one can use Banach's fixed point
theorem to conclude that ${\T}$ has a unique fixed point $b$ for which $\T b=b$. 
The sequence $b=(b_n)_{n \geq 0}$ is positive and it is a solution of the
alternative discrete Painlev\'e equation \eqref{dPI-b}.
The contraction factor can be improved by improving the upper and the lower bounds in Lemmas \ref{lem21} and \ref{lem22}. 
The lower bound in Lemma \ref{lem22} can be used to get a better upper bound in Lemma \ref{lem21}, which in turn can be used to improve the lower bound in Lemma \ref{lem22}. 

We have not yet shown that $b_0(t)=-\Ai'(t)/\Ai(t)$; %this will be done in \S\ref{Sec5}, 
see Corollary \ref{cor52}.

\section{\label{Sec3}Behaviour of the unique positive solution}
The unique positive solution of \eqref{dPI-b} necessarily satisfies $b_n(t)>\sqrt{t}$, which is an immediate consequence of \eqref{dPI-b}. 
Furthermore, Lemma \ref{lem21} implies 
$\sqrt{t}<b_n(t)\leq B_n(t)$, for any $n\geq0$ and $t\geq0$.
This is illustrated in Figure \ref{plotbB} for $n=5$ and $n=10$.
\begin{figure}[ht]
%\centerline{\includegraphics[scale=0.67]{B6B7b6b7.jpg}}
\[\begin{array}{cc}
\includegraphics[width=3in,height=2.5in]{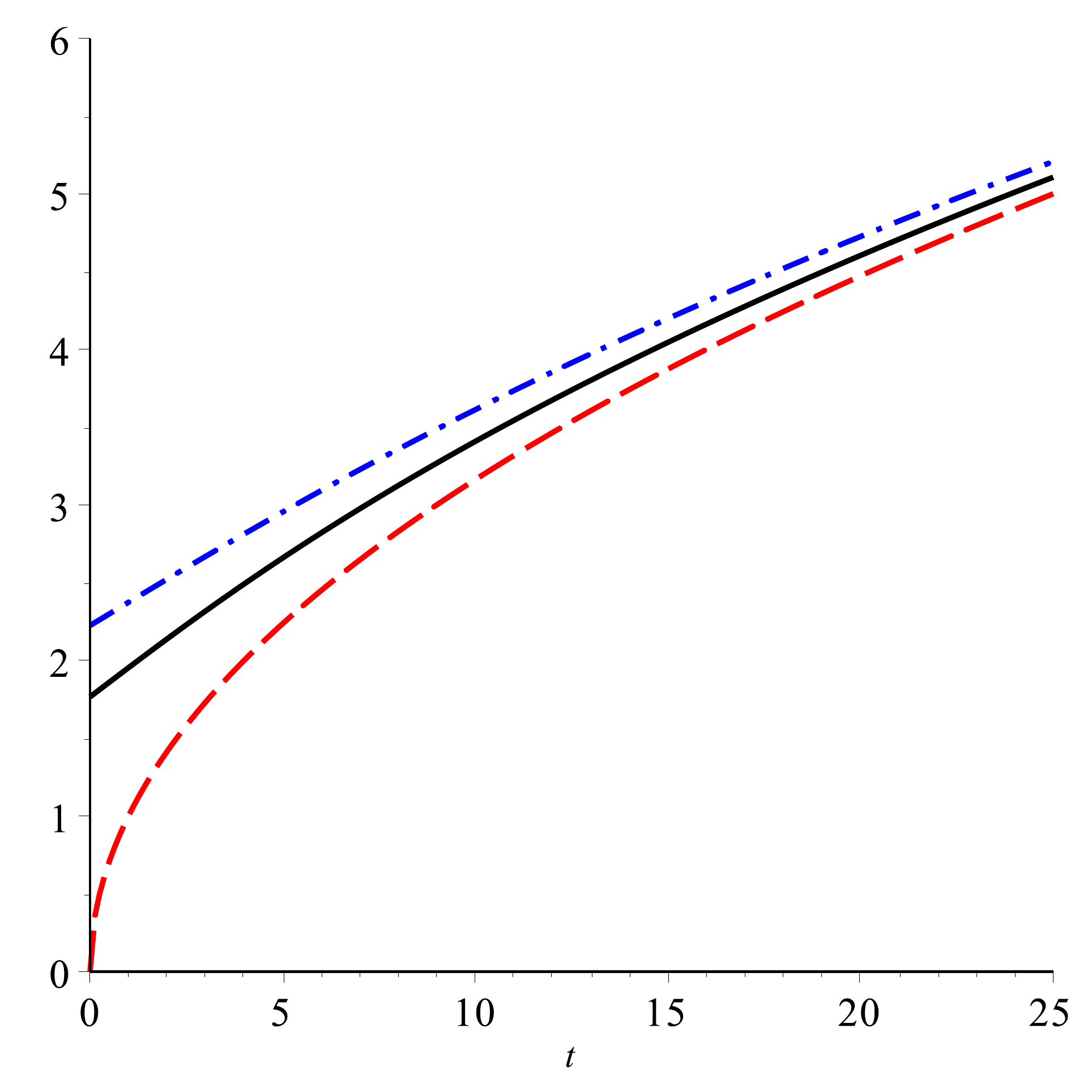} &\includegraphics[width=3in,height=2.5in]{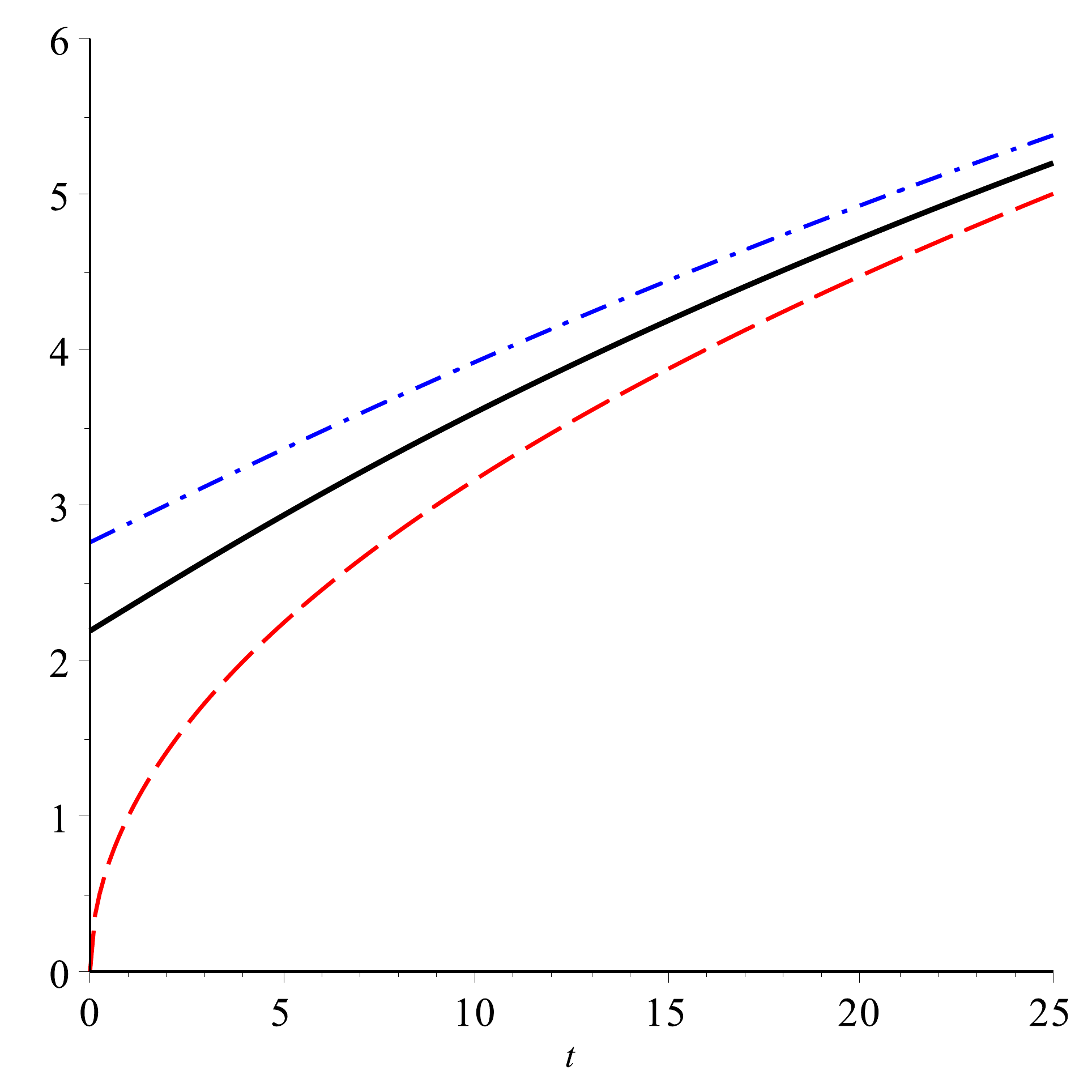}\\
B_5(t), b_5(t), \sqrt{t} & B_{10}(t), b_{10}(t), \sqrt{t}
\end{array}\]
\caption{\label{plotbB}Plots of $B_n(t)$ (dot-dash line), $b_n(t)$ (solid line) and $\sqrt{t}$ (dashed line) for $n=5$ and $n=10$.}
\end{figure}
Moreover, recalling the definition of $B_n(t)$ given by \eqref{Rnt}, we observe that 
$\lim_{t\to \infty} {B_n(t)}/{\sqrt{t}} = 1$,
and consequently, 
\begin{equation}\label{lim bn large t}
	\lim_{t\to \infty} {b_n(t)}/{\sqrt{t}} = 1. 
\end{equation}

\begin{remark} Observe that for any $t\geq 0$ and any integer $n\geq 0$, the upper bound $B_n(t)$ for $b_n(t)$ given in Lemma \ref{lem21} satisfies 
$\sqrt{t} < B_n(t) < \sqrt{t+(2n+1)^{2/3}}$.
For values of $t>3\left( n+\tfrac{1}{2} \right)^{2/3}$, $B_n(t)$ can be written as 
$$B_n(t)= \tfrac{2}{{3}}\sqrt{3\,t}
		\cos\left\{ \tfrac{1}{3} \arccos\left[{3^{3/2}( n+\tfrac{1}{2})}/{t^{3/2}}\right]\right\},
		%\quad \text{with } \quad z=\frac{t}{3\left( n+\tfrac{1}{2} \right)^{2/3}}
$$ for any $n\geq0$. Since $0\leq\tfrac{1}{3} \arccos\left[{3^{3/2}( n+\frac{1}{2})}/{t^{3/2}}\right]<\tfrac16\pi$ for $t>3\left( n+\tfrac{1}{2} \right)^{2/3}$, it follows that $\sqrt{t} < B_n(t)\leq \tfrac{2}{3}\sqrt{3\,t}$.
\end{remark}

\begin{remark} The B\"acklund transformations pair 
\begin{subequations}\label{dPI-b transf}\begin{align} 
 b_n(t) + b_{n+1}(t) = \frac{2(n+1)}{b_n^2(t) - t + b_n'(t)} , \label{dPI-b transf-1} \\
 b_n(t) + b_{n-1}(t) = \frac{2n}{b_n^2(t) - t - b_n'(t)} , \label{dPI-b transf-2}
\end{align}\end{subequations}
relate solutions of the equation \eqref{dPI-b}. In fact, \eqref{dPI-b} can be obtained by eliminating $b_n'(t)$ between \eqref{dPI-b transf-1} and \eqref{dPI-b transf-2}. 
%
%$$
%	 b_n'(t) = \frac{2(n+1)}{ b_n(t) + b_{n+1}(t)}- (b_n^2(t) - t)
%$$
%$$
%	 b_n'(t)= b_n^2(t) - t - \frac{2n}{ b_n(t) + b_{n-1}(t)} 
%$$
Hence, 
$$
	b_n'(t) = \frac{n+1}{ b_n(t) + b_{n+1}(t)} - \frac{n}{ b_n(t) + b_{n-1}(t)}, 
$$
which, recalling \eqref{dPI-2}, corresponds to 
\begin{equation}\label{bn prime}
	b_n'(t) = a_{n+1}(t)-a_n(t). 
\end{equation}
Moreover, \eqref{dPI-1} and the latter equality yield 
\begin{equation}\label{diff bn bn prime}
	b_{n+1}^2(t) - b_{n}^2(t) = b_{n+1}'(t) + b_n'(t). 
\end{equation}
\end{remark}

The identities \eqref{bn prime} and \eqref{diff bn bn prime} allow us to conclude the following result. 

\begin{lemma} For each integer $n\geq0$, $b_n'(t)>0$ for all $t\geq0$ if and only if $a_{n+1}(t)>a_{n}(t)$, and this implies $b_{n+1}(t)>b_{n}(t)$ and also $a_{n}'(t)<0$. 
\end{lemma}

Numerically it is evident that $b_n(t)$ is a strictly increasing function of $t$. However, it remains an open problem to prove this (analytically). 

\begin{conjecture}\label{conj32} For any $t\geq 0$, $ 0<b_n'(t) < 1/(2\sqrt{t}) $. 
\end{conjecture}

Notwithstanding this, we can still prove the result for values of $t$ larger than $2^{-2/3}n^{4/3}$. 

\begin{lemma}\label{lem33} For each integer $n\geq 0$, we have $b_n'(t)>0$ for any $t \geq \left(\frac{(2n+3)n^3}{(2n+1) (n+1)}\right)^{2/3}$. 
\end{lemma}
\begin{proof} Since $a_0=0$, then \eqref{bn prime} readily implies 
$b_0'(t) = a_1(t) > 0$,
for any $t\geq0$. 
Based on \eqref{bn prime} together with the inequalities $$\frac{n}{2\sqrt{t}} < a_n < \frac{n}{B_{n}(t) + B_{n-1}(t)},$$ which are valid for any $n\geq1$ and $t\geq0$, it follows that 
$$
	b_n'(t) > \frac{n+1}{B_{n+1}(t) + B_n(t)} - \frac{n}{2\sqrt{t}}. 
$$
Since $B_n(t)$ is an increasing function of both $n$ and $t$, bounded from below by $\sqrt{t}$ and $\lim\limits_{t\to+\infty} \frac{B_n(t)}{\sqrt{t}}=1$, it follows that for each $n\geq 1$ there exists $t^*_n$ such that for any $t> t_n^*$, we have 
${ B_n(t)+B_{n+1}(t) \leq \left(1+\frac{1}{n}\right) 2\sqrt{t} }$, which implies $b_n'(t) >0.$

Indeed, let $t_n^*$ be such that $B_{n+1}\left(t_n^*\right) = \frac{n+1}{n} \sqrt{t_n^*}$. As $B_n(t)$ is by definition the positive solution of $y (y^2-t) = 2n+1$, it readily follows that $t_n^*$ must be such that 
$$\frac{n+1}{n^3} \left[(n+1)^2 - n^2 \right] %{t_n^*}^{\;3/2}
\big(\,t_n^*\big)^{3/2} 
= 2n+3, $$
and hence $ t_n^* = \left[\frac{(2n+3)n^3}{(2n+1) (n+1)}\right]^{2/3}$. 
Then for $t>t_n^*$, $ B_n(t)+B_{n+1}(t) \leq 2 B_{n+1}(t)$, for each $n\geq 0$.
\end{proof}

\comment{\includegraphics[scale=0.7]{pictureProofbnderivative.jpg}
\begin{figure}[ht]
\[\begin{array}{cc}
\includegraphics[width=3in,height=2.5in]{altdPIB8} &\includegraphics[width=3in,height=2.5in]{altdPIB10}\\
B_8(t), \tfrac87\sqrt{t} & B_{10}(t),\tfrac{10}9 \sqrt{t}\\
\includegraphics[width=3in,height=2.5in]{altdPIB12} &\includegraphics[width=3in,height=2.5in]{altdPIB15}\\
B_{12}(t), \tfrac{12}{11}\sqrt{t} & B_{15}(t),\tfrac{15}{14} \sqrt{t}
\end{array}\]
\caption{\label{plotBn}Plots of $B_{n+1}(t)$ (solid line) and $\frac{n+1}{n}\sqrt{t}$ (dashed line) for various $n$. %$n=7,8,9,10$. 
These illustrate intersection points at $\widehat{t}_8=13.316$, %$\widehat{t}_9=15.930$, $\widehat{t}_{11}=21.482$
$\widehat{t}_{10}=18.655$, $\widehat{t}_{12}=24.405$ and $\widehat{t}_{15}=33.690$.}
\end{figure}}

 In the same manner we can show the existence of $t^+_n$ such that for any $t>t^+_n$, we have $ b_n'(t)< {1}/({2\sqrt{t}})$. 

\begin{remark} The asymptotic behaviour of $b_n$ for large $t$ in \eqref{lim bn large t} allows us to conclude from \eqref{dPI-2} and \eqref{bn prime}, respectively, that 
$\lim_{t\to +\infty} \left(a_n(t) \, {2\sqrt{t} }/{n} \right)= 1$ and
%$$
%whilst \eqref{bn prime} implies $$
$\lim_{t\to +\infty}\left(b_n'(t) \, 2\sqrt{t} \right)= 1$.
\end{remark}

\begin{remark} 
As $b_n(t)>\sqrt{t}$ for any $n\geq 0$, then \eqref{dPI-b transf-1} implies 
$$
	0< b_n(t) + b_{n+1}(t) = \frac{2(n+1)}{b_n^2(t) - t + b_n'(t)} 
$$
whereas \eqref{dPI-b transf-2} implies that 
$$
	0< b_n(t) + b_{n-1}(t) = \frac{2n}{b_n^2(t) - t - b_n'(t)} . 
$$
Hence, we have 
$$
	\Big[b_n^2(t) - t + b_n'(t)\Big]\Big[b_n^2(t) - t -b_n'(t) \Big]>0, 
$$
{\it i.e.},
$-(b_n^2(t) - t ) <b_n'(t) <b_n^2(t) - t $. 
\end{remark}

\section{\label{Sec4}Relation with the second Painlev\'e equation}
%The second Painlev\'e equation %\PII\ \eqref{eqPII} 
%is \begin{equation} \label{eqPII}\deriv[2]{w}{z} = 2w^3 + zw + \alpha 
%\end{equation} and it 
%has solutions which we denote by $w(z;\alpha)$. 
Suppose we denote solutions of \PII\ \eqref{eqPII}  by $w(z;\alpha)$, then the
three solutions $w(z;\alpha)$ and $w(z;\alpha\pm1)$ are related by
\begin{equation}\label{p2rr} \frac{\alpha+\frac12}{w(z;\alpha+1)+w(z;\alpha)} + \frac{\alpha-\frac12}{w(z;\alpha)+w(z;\alpha-1)} 
+2w^2(z;\alpha)+z=0, 
\end{equation}
see \cite[Eq.~32.7.5 on p.~730]{NIST} and \cite[Eq.~(1.21)]{refFokasGR}. The latter is equivalent to the difference equation \eqref{dPI-b}.
If we set $\alpha=n+\frac12$ \comment{and $y_n(t)=cw(ct;n+\frac12)$ then 
\begin{equation}\label{p2scalrr} 
\frac{n+1}{y_{n+1} + y_n} + \frac{n}{y_{n}+y_{n-1}} + 2c^{-3} y_n^2 + t = 0,\end{equation}
hence the choice $c=-2^{1/3}$ gives \eqref{dPI-b} so that}%
then $b_n(t) = -2^{1/3}w(-2^{1/3}t;n+\frac12)$.
It is known that for $\alpha=n+\frac12$ then \PII\ \eqref{eqPII} has solutions in terms of Airy functions (see \cite[32.10(ii) on p.~735]{NIST} or \cite[\S 7.1 on p.~373]{Clarkson}). The simplest ``Airy solution" is 
$$w(z;\tfrac12)=-\deriv{}{z}\ln\phi(z)=-\frac{\phi'(z)}{\phi(z)},$$ 
where
\begin{equation}\label{eq:phi} \phi(z) = C_1 \Ai(-2^{-1/3}z) + C_2 \Bi(-2^{-1/3}z),\end{equation}%\qquad\zeta = -2^{-1/3}z \]
with $\Ai(t)$ and $\Bi(t)$ the Airy functions and $C_1$ and $C_2$ arbitrary constants. 
Observe that the solution $w(z;\tfrac12)$ only depends on the ratio $C_1/C_2$. We now have that 
\[ y_0(t) = -2^{-1/3}w(-2^{1/3}t;\tfrac12) = 2^{-1/3} \frac{\phi'(-2^{1/3}t)}{\phi(-2^{1/3}t)}, \]
and since $\phi(-2^{-1/3}t) = C_1 \Ai(t)+ C_2 \Bi(t)$, 
%\phi'(-2^{-1/3}t) = -2^{1/3} \bigl[ C_1 \Ai'(t) + C_2 \Bi'(t) \bigr], \end{split}\]
then we find
\begin{equation}\label{y0Airy}
y_0(t) = -\frac{C_1 \Ai'(t) + C_2 \Bi'(t)}{C_1 \Ai(t) + C_2 \Bi(t) } . \end{equation}
For $t=0$ this gives
\[ b_0(0) = -\frac{C_1 \Ai'(0) + C_2 \Bi'(0)}{C_1 \Ai(0) + C_2 \Bi(0)} = 
 \frac{3[\Gamma(\tfrac23)]^2}{2\pi}\left( \frac{C_1 - \sqrt{3}\, C_2}{C_1+\sqrt{3}\,C_2}\right), \]
where we used the initial values
\[ \Ai(0) = \frac{1}{3^{2/3} \Gamma(\tfrac23)}, \quad\Bi(0) = \frac{1}{3^{1/6} \Gamma(\tfrac23)}, \quad
 \Ai'(0) = -\frac{3^{1/6} \Gamma(\tfrac23)}{2\pi}, \quad\Bi'(0) = \frac{3^{2/3}\Gamma(\tfrac13)}{2\pi}, \]
%\[ \Ai(0) = \frac{3^{-2/3}}{\Gamma(\tfrac23)}, \quad\Bi(0) = \frac{3^{-1/6}}{ \Gamma(\tfrac23)}, \quad
% \Ai'(0) = -\frac{3^{-1/3} }{\Gamma(\tfrac13)}, \quad\Bi'(0) = \frac{3^{1/6}}{ \Gamma(\tfrac13)}, \]
see \cite[9.2(ii) on p.~194]{NIST}. Thus choosing $C_2=0$ gives the initial value mentioned in Theorems \ref{thmab} and \ref{thmb} for $t=0$.

More generally, we have the solution
\begin{equation} w(z;n+\tfrac12)=\deriv{}{z}\ln\frac{\Theta_n(z)}{\Theta_{n+1}(z)},\qquad
%where for $n\geq1$ $$
\Theta_n(z)=\det\left[\deriv[j+k]{}{z}\phi(z)\right]_{j,k=0}^{n-1},\end{equation}
%=\mathcal{W}\left(\phi,\phi',\ldots,\phi^{(n-1)}\right),\qquad \phi^{(m)}=\deriv[m]{}{z}\phi(z)$$
for $n\geq1$, with %$\Theta_0(z)=1$ and 
$\phi(z)$ as given by \eqref{eq:phi}, cf.~\cite{Clarkson,ForrW01,refOkamotoPIIPIV}.

%\subsection{Airy solutions}
Consider the B\"acklund transformations
\begin{subequations}\label{p2scbt}\begin{align}
y_{n+1} &= -y_n+ \frac{2(n+1)}{y_n^2+y_n'-t},\label{p2scbta}\\
y_{n-1} &= -y_n+ \frac{2n}{y_n^2-y_n'-t},\label{p2scbtb}
\end{align}\end{subequations}
where $'\equiv d/dt$.
Eliminating $y_n'$ yields the recurrence relation %\eqref{p2scrr}, 
\begin{equation}\label{p2scrr} \frac{n+1}{y_{n+1} + y_n} + \frac{n}{y_{n}+y_{n-1}} =y_n^2 -t ,\end{equation}
which is alt-\dPI\ \eqref{dPI-b},
whilst letting $n\to n+1$ in \eqref{p2scbtb} and then substituting \eqref{p2scbta} yields 
\begin{equation} \label{PIIscn}
\deriv[2]{y_n}{t} = 2y_n^3 - 2ty_n - 2n-1,
\end{equation}
which is \eqref{PIIsc} with $\alpha=n+\tfrac12$,
and so is equivalent to \PII\ \eqref{eqPII} with $\alpha=n+\tfrac12$.

\begin{lemma}{\label{lem41}Suppose $x_n$ and $y_n$ satisfy the discrete system
\begin{subequations}\label{dPIsysxy}\begin{align} 
& x_n+x_{n+1} = y_n^2 -t , \label{dPI-1xy} \\
 & x_n(y_n+y_{n-1}) = n, \label{dPI-2xy}
\end{align}\end{subequations}
and $y_n$ satisfies \eqref{p2scbt}. Then  $x_n$ and $y_n$ satisfy the system
\begin{subequations}\label{PIIsysxy}\begin{align} 
\deriv{x_n}{t}&=-2x_ny_n+n,\label{PIIsysx}\\
\deriv{y_n}{t}&=y_n^2-2x_n-t,\label{PIIsysy}
\end{align}\end{subequations}
and $x_n$ satisfies
\begin{equation}
\deriv[2]{x_n}{t}=\frac{1}{2x_n}\left(\deriv{x_n}{t}\right)^2+4x_n^2+2tx_n-\frac{n^2}{2x_n}.\label{P34sc}
\end{equation}}\end{lemma}

\begin{proof}
From \eqref{dPI-2xy} and \eqref{p2scbtb} we have
\[\frac{2n}{y_n+y_{n-1}}=2x_n=y_n^2-\deriv{y_n}{t}-t,\]
from which we obtain equation \eqref{PIIsysy}. Then differentiating \eqref{PIIsysy} gives
\[\deriv[2]{y_n}{t}=2y_n\deriv{y_n}{t}-2\deriv{x_n}{t}-1.\]
Substituting for the derivatives of $y_n$ using \eqref{PIIsysy} and \eqref{PIIscn} yields equation \eqref{PIIsysx}. Finally solving  \eqref{PIIsysx} for $y_n$ and substituting in \eqref{PIIsysy} shows that $x_n$ satisfies equation \eqref{P34sc}, as required.
\end{proof}

We remark that making the transformation
$x_n(t)=-2^{-1/3}v(z)$, with $z=-2^{1/3}t$,
in equation \eqref{P34sc} yields
\begin{equation}\label{eqp34}
\deriv[2]{v}{z}=\frac{1}{2v}\left(\deriv{v}{z}\right)^2-2v^2-zv-\frac{n^2}{2v},\end{equation} 
which is known as $\mbox{P}_{\!34}$, %(cf.~\cite{Clarkson}), 
as it's equivalent to equation XXXIV of Chapter 14 in \cite{refInce}.
%which itself is equivalent to \PII\ since there is a one-to-one relationship between solutions of \eqref{eqp34} and those of \PII. 

The ``Airy-type" solutions of \eqref{PIIscn} and \eqref{P34sc} respectively have the form
%\begin{subequations}
\begin{align}
y_n(t;\th)&=\deriv{}{t}\ln \frac{\tau_n(t;\th)}{\tau_{n+1}(t;\th)},\qquad %\label{solyn}\\
\label{solyxn}
x_n(t;\th)=-\deriv[2]{}{t}\ln{\tau_n(t;\th)},
\end{align}%\end{subequations}
where %$\tau_n(t;\th)$ is the Wronskian
\begin{equation}
\tau_n(t;\th)=\det\left[\deriv[j+k]{}{t}\varphi(t;\th)\right]_{j,k=0}^{n-1},%\qquad \tau_0(t;\th)=1,
%=\mathcal{W}\left(\varphi,\varphi',\ldots,\varphi^{(n-1)}\right),\qquad \varphi^{(m)}=\deriv[m]{\varphi}{t},
\end{equation}with $\tau_0(t;\th)=1$
and
%\begin{equation}
$\varphi(t;\th)=\cos(\th)\Ai(t)+\sin(\th)\Bi(t)$,
%\end{equation}
with $\Ai(t)$ and $\Bi(t)$ the Airy functions and $\th\in[0,\pi)$ a parameter; we have set $C_1=\cos(\th)$ and $C_2=\sin(\th)$ to reflect that the solution only depends on the ratio of the constants. The ``Airy-type" solutions \eqref{solyxn} are derived from the ``Airy-type" solutions of \PII\ \eqref{eqPII} given in \cite{Clarkson,ForrW01,refOkamotoPIIPIV}.

The simplest non-trivial ``Airy-type" solutions of \eqref{PIIscn} and \eqref{P34sc} respectively have the form
\begin{subequations}\begin{align} y_0(t;\th)&=-\deriv{}{t}\ln\varphi(t;\th)%=-\frac{\varphi'(t;\th)}{\varphi(t;\th)}
=-\frac{\cos(\th)\Ai'(t)+\sin(\th)\Bi'(t)}{\cos(\th)\Ai(t)+\sin(\th)\Bi(t)},\label{soly0}\\
x_1(t;\th)&=-\deriv[2]{}{t}\ln\varphi(t;\th)%=\left[\frac{\varphi'(t;\th)}{\varphi(t;\th)}\right]^2-t\nonumber\\
=\left[\frac{\cos(\th)\Ai'(t)+\sin(\th)\Bi'(t)}{\cos(\th)\Ai(t)+\sin(\th)\Bi(t)}\right]^2-t,\label{solx1}\end{align}\end{subequations}
recall that $x_0(t;\th)=0$.
The structure of the solutions \eqref{soly0} and \eqref{solx1} depends critically on whether the parameter $\th$ is zero or not which is shown in Lemmas \ref{lem42} and \ref{lem42x} below. As $\Ai(t)$ decays exponentially as $t\to\infty$, whereas $\Bi(t)$ increases exponentially as $t\to\infty$, so if $\th\not=0$ then $\Bi(t)$ will dominate for large positive $t$.

If we seek a solution of \eqref{PIIscn} with $y_n(t)\sim c\,t^{1/2}$, as $t\to\infty$, then necessarily $c=\pm1$.
%with leading behaviour $y_n(t)\sim\pm\sqrt{t}$. 
Hence the following asymptotic series are easily derived
\begin{subequations}\label{solynab}
\begin{align} \label{solyna} 
y_{n}^{+}(t)&=t^{1/2}+{\frac{2n+1}{4\,t}}- \frac{12n^2+12n+5}{32\,t^{5/2}}+\O\big(t^{-4}\big),\\
%\frac{(2n+1)(16n^2+16n+15)}{64\,t^{4}}+\O\big(t^{-11/2}\big),\\ 
\label{solynb} 
y_{n}^{-}(t)&=-t^{1/2}+{\frac{2n+1}{4\,t}}+ \frac{12n^2+12n+5}{32\,t^{5/2}}+\O\big(t^{-4}\big).
%+\frac{(2n+1)(16n^2+16n+15)}{64\,t^{4}}+\O\big(t^{-11/2}\big).
\end{align}\end{subequations}
%\[y_n^{\pm}(t)=\pm t^{1/2}+{\frac{2n+1}{{4t}}}\mp \frac{12n^2+12n+5}{32\,t^{5/2}}+\frac{(2n+1)(16n^2+16n+15)}{64\,t^{4}}+\O\big(t^{-11/2}\big).\] 
It should be noted that there are no arbitrary constants in these asymptotic expansions, they occur in exponentially small terms.

\begin{lemma}\label{lem42} If $y_0(t;\th)$ is given by \eqref{soly0}, then
\begin{equation} y_0(t;\th)=\begin{cases} t^{1/2}+ \O\big(t^{-1}\big), &\quad{\rm if}\quad\th=0,\\
-t^{1/2}+ \O\big(t^{-1}\big), &\quad{\rm if}\quad\th\not=0.\end{cases}\end{equation}
\end{lemma}

\begin{proof}Using the known asymptotics of $\Ai(t)$ and $\Bi(t)$
\begin{align*}
\Ai(t)&=\frac{\e^{-\zeta}}{2\sqrt{\pi}\,t^{1/4}}\left\{1-\frac{5}{48\,t^{3/2}}+\frac{385}{4608\,t^3}-\frac{85085}{663552\,t^{9/2}} + \O\big(t^{-6}\big)\right\}\\
\Ai'(t)&=-\frac{t^{1/4}\,\e^{-\zeta}}{2\sqrt{\pi}}\left\{1+\frac{7}{48\,t^{3/2}}+\frac{455}{4608\,t^3}+\frac{95095}{663552\,t^{9/2}} + \O\big(t^{-6}\big)\right\}\\
\Bi(t)&=\frac{\e^{\zeta}}{\sqrt{\pi}\,t^{1/4}}\left\{1+\frac{5}{48\,t^{3/2}}+\frac{385}{4608\,t^3}+\frac{85085}{663552\,t^{9/2}} + \O\big(t^{-6}\big)\right\}\\
\Bi'(t)&=\frac{t^{1/4}\,\e^{\zeta}}{\sqrt{\pi}}\left\{1-\frac{7}{48\,t^{3/2}}-\frac{455}{4608\,t^3}-\frac{95095}{663552\,t^{9/2}} + \O\big(t^{-6}\big)\right\}
\end{align*}
%\begin{align*}
%&1-{\frac{5}{48\,t^{3/2}}}+{\frac{385}{4608\,t^3}}-{\frac{85085}{663552\,t^{9/2}}}\\
%&1+{\frac{7}{48\,t^{3/2}}}-{\frac{455}{4608\,t^3}}+{\frac{95095}{663552\,t^{9/2}}}\\
%&1+{\frac{5}{48\,t^{3/2}}}+{\frac{385}{4608\,t^3}}+{\frac{85085}{663552\,t^{9/2}}}\\
%&1-{\frac{7}{48\,t^{3/2}}}-{\frac{455}{4608\,t^3}}-{\frac{95095}{663552\,t^{9/2}}}\end{align*}
with $\zeta=\tfrac23t^{3/2}$, see \cite[9.7(ii) on p.~198]{NIST}, then if $\th=0$, as $t\to\infty$
\begin{align*}
y_0(t;0)&=-\frac{\Ai'(t)}{\Ai(t)}= t^{1/2}\frac{\displaystyle\left\{1+\frac{7}{48\,t^{3/2}}+\frac{455}{4608\,t^3}+\frac{95095}{663552\,t^{9/2}} + \O\big(t^{-6}\big)\right\}}{\displaystyle\left\{1-\frac{5}{48\,t^{3/2}}+\frac{385}{4608\,t^3}-\frac{85085}{663552\,t^{9/2}} + \O\big(t^{-6}\big)\right\}} \\ %&=t^{1/2}\left\{1+\frac{1}{4\,t^{3/2}}-{\frac{5}{32\,t^3}}+{\frac{15}{64\,t^{9/2}}}+ \O\big(t^{-6}\big)\right\}\\
&=t^{1/2}+\frac{1}{4\,t}-{\frac{5}{32\,t^{5/2}}}+{\frac{15}{64\,t^{4}}} + \O\big(t^{-11/2}\big)
\end{align*}
whilst if $\th\not=0$, as $t\to\infty$
\begin{align*}
y_0(t;\th)&=-\frac{\cos(\th)\Ai'(t)+\sin(\th)\Bi'(t)}{\cos(\th)\Ai(t)+\sin(\th)\Bi(t)}\\
&=-t^{1/2}\frac{\displaystyle \tfrac12\cos(\th)\,\e^{-\zeta}\left\{1+\frac{7}{48\,t^{3/2}}+ \O\big(t^{-3}\big)
%+\frac{455}{4608\,t^3}+\frac{95095}{663552\,t^{9/2}} + \O\big(t^{-6}\big)
\right\}+\sin(\th)\,\e^{\zeta}\left\{1-\frac{7}{48\,t^{3/2}}+ \O\big(t^{-3}\big)
%-\frac{455}{4608\,t^3}-\frac{95095}{663552\,t^{9/2}} + \O\big(t^{-6}\big)
\right\} }
{\displaystyle \tfrac12\cos(\th)\,\e^{-\zeta}\left\{1-\frac{5}{48\,t^{3/2}}+ \O\big(t^{-3}\big)
%+\frac{385}{4608\,t^3}-\frac{85085}{663552\,t^{9/2}} + \O\big(t^{-6}\big)
\right\}+\sin(\th)\,\e^{\zeta}\left\{1+\frac{5}{48\,t^{3/2}}+ \O\big(t^{-3}\big)
%+\frac{385}{4608\,t^3}+\frac{85085}{663552\,t^{9/2}} + \O\big(t^{-6}\big)
\right\}}\\
%&\sim-\frac{\Bi'(t)}{\Bi(t)}+\O\left(t^{1/2}\,\e^{-\zeta}\right)\\
&= -t^{1/2}\frac{\displaystyle\left\{1-\frac{7}{48\,t^{3/2}}-\frac{455}{4608\,t^3}-\frac{95095}{663552\,t^{9/2}} + \O\big(t^{-6}\big)\right\}}{\displaystyle\left\{1+\frac{5}{48\,t^{3/2}}+\frac{385}{4608\,t^3}+\frac{85085}{663552\,t^{9/2}} + \O\big(t^{-6}\big)\right\}}+\O\left(t^{1/2}\,\e^{-2\zeta}\right)\\ 
%&=-t^{1/2}\left\{1-\frac{1}{4\,t^{3/2}}-{\frac{5}{32\,t^3}}-{\frac{15}{64\,t^{9/2}}}+\O\big(t^{-6}\big)\right\}\\
&=-t^{1/2}+\frac{1}{4\,t}+{\frac{5}{32\,t^{5/2}}}+{\frac{15}{64\,t^{4}}}+ \O\big(t^{-11/2}\big)
\end{align*}
%Hence we have proved the following
\end{proof}
Plots of $y_0(t;\th)$ and $x_1(t;\th)$ for various values of the parameter $\th$ are given in Figure \ref{ploty0x1}. We note that
$$y_0(0;\th)=-\frac{\cos(\th)\Ai'(0)+\sin(\th)\Bi'(0)}{\cos(\th)\Ai(0)+\sin(\th)\Bi(0)}
%=\frac{3^{5/6}\big(1-\sqrt{3}\,\tan\th\big)}{2\pi\big(1+\sqrt{3}\,\tan\th\big)}\big[\Gamma(\tfrac23)\big]^2,
=\frac{3^{5/6}}{2\pi}\left(\frac{1-\sqrt{3}\,\tan\th}{1+\sqrt{3}\,\tan\th}\right)\big[\Gamma(\tfrac23)\big]^2,
$$
which is a decreasing function for $0\leq\th\leq\tfrac12\pi$.
{\begin{figure}[ht]
\[\begin{array}{cc}
\includegraphics[width=3in,height=2.5in]{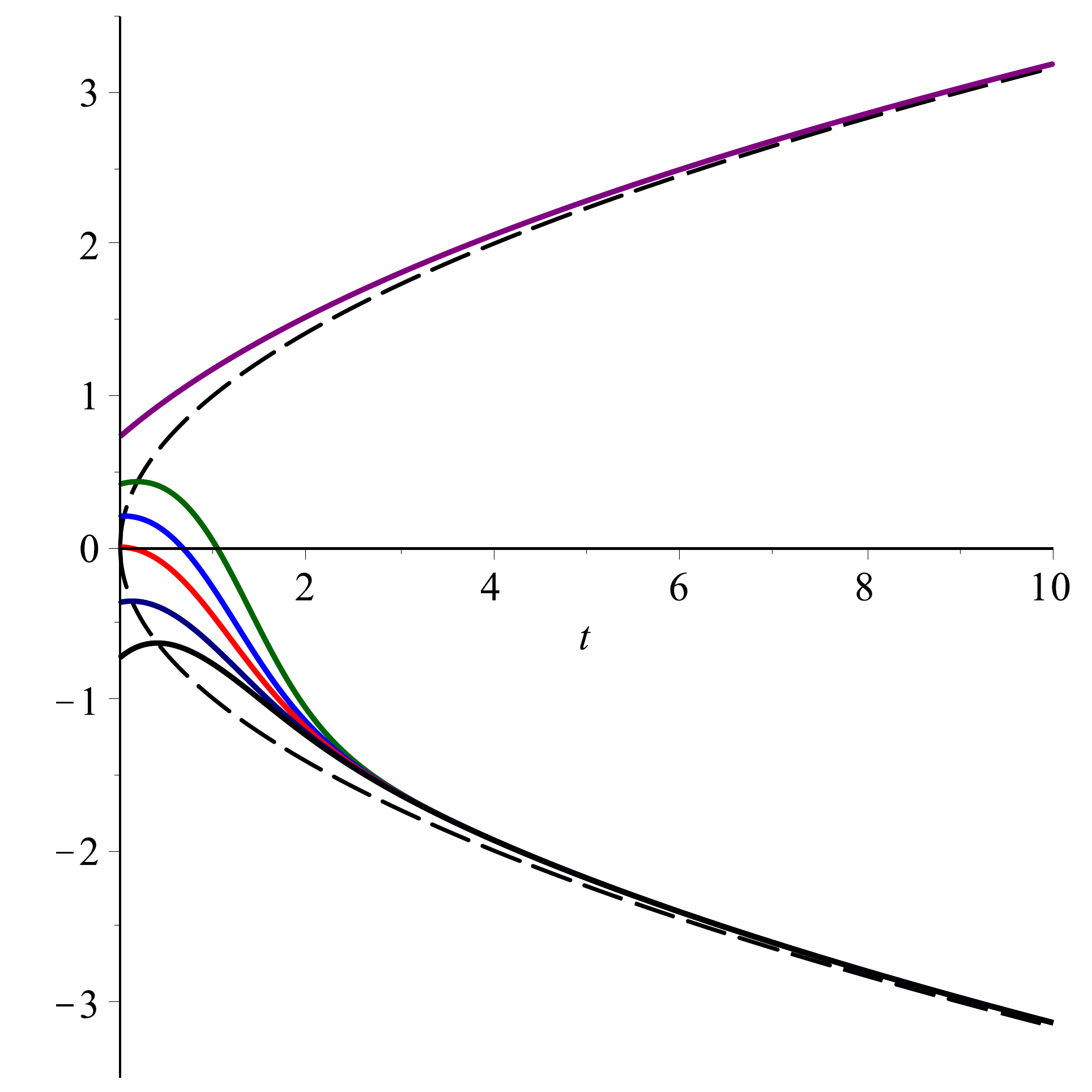} & \includegraphics[width=3in,height=2.5in]{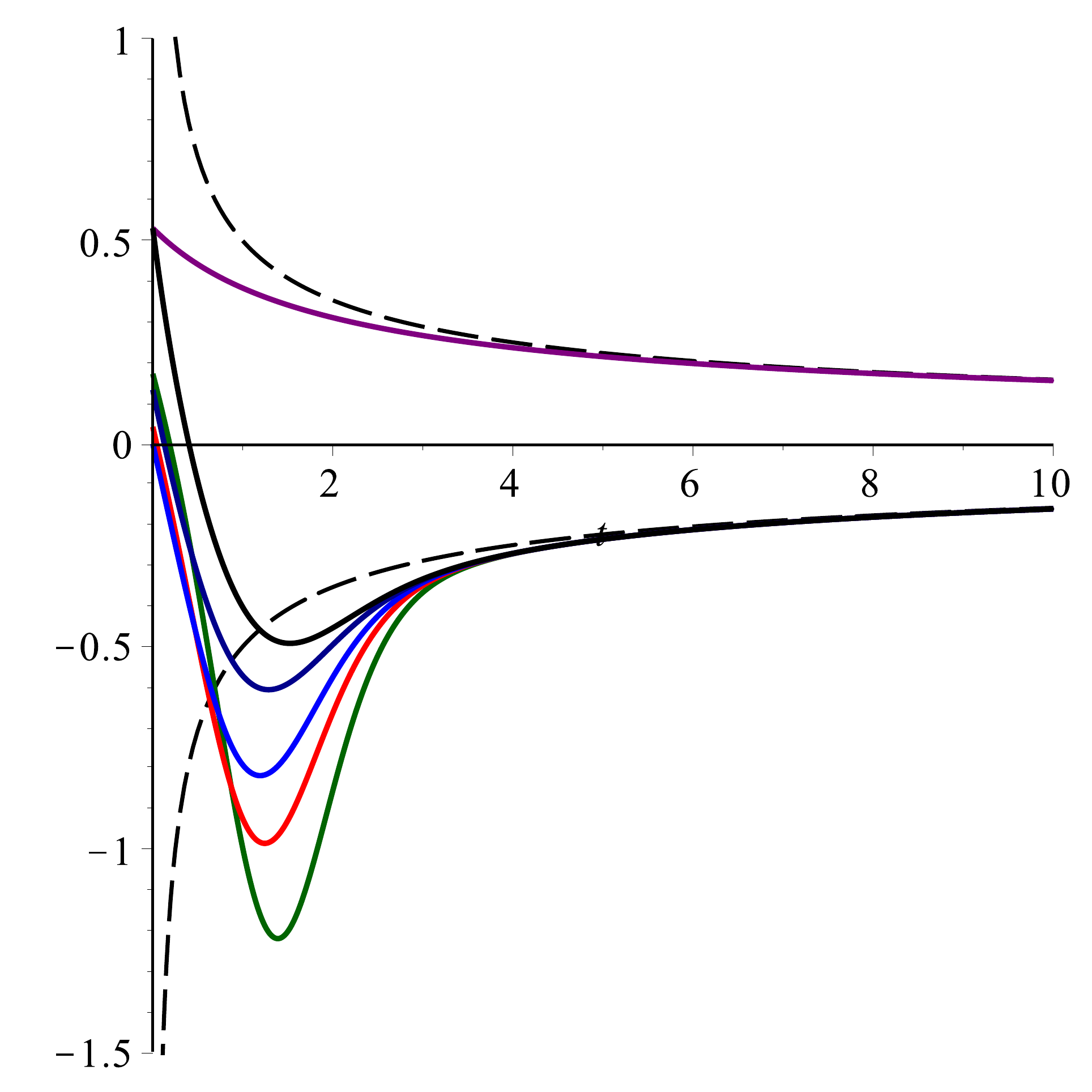} \\
{\rm (i)} & {\rm (ii)}  \end{array} \]
\caption{\label{ploty0x1}
(i) Plots of $y_0(t;\th)$ for $\th=0$ (upper curve), $\th=\tfrac1{20}\pi,\tfrac1{10}\pi,\tfrac16\pi,\tfrac13\pi,\tfrac12\pi$ (lower curves, with $\th=\tfrac1{20}\pi$ the highest and $\th=\tfrac1{2}\pi$ the lowest), and the %parabola $y^2-t=0$ 
curves $y=\pm\sqrt{t}$\ (dashed lines).%} %,\tfrac34\pi
%\end{figure}\begin{figure}[ht]\[\includegraphics[width=3in]{altdP1x1theta1a} \]
%\caption{\label{plotx1}
(ii) Plots of $x_1(t;\th)$ for $\th=0$ (upper curve), $\th=\tfrac1{20}\pi,\tfrac1{10}\pi,\tfrac16\pi,\tfrac13\pi,\tfrac12\pi$ (lower curves, with $\th=\tfrac1{20}\pi$ the lowest and $\th=\tfrac1{2}\pi$ the highest), and the curves $x=\pm 1/(2\sqrt{t})$ %curve $4tx^2=1$ 
(dashed lines).} %,\tfrac34\pi
\end{figure}}
\comment{\begin{figure}[ht]
\[\includegraphics[width=3in]{altdP1y0theta1b} \]
\caption{\label{ploty0}Plots of $y_0(t;\th)$ for $\th=0$ (upper curve), $\th=\tfrac1{1000}\pi,\tfrac1{100}\pi,\tfrac1{25}\pi,\tfrac1{10}\pi,\tfrac15\pi,\tfrac12\pi$ (lower curves, with $\th=\tfrac1{20}\pi$ the highest and $\th=\tfrac1{2}\pi$ the lowest), and the %parabola $y^2-t=0$ 
curves $y=\pm\sqrt{t}$\ (dashed lines).} %,\tfrac34\pi
\end{figure}
\begin{figure}[ht]
\[\includegraphics[width=3in]{altdP1x1theta1b} \]
\caption{\label{plotx1}Plots of $x_1(t;\th)$ for $\th=0$ (upper curve), $\th=\tfrac1{1000}\pi,\tfrac1{100}\pi,\tfrac1{25}\pi,\tfrac1{10}\pi,\tfrac15\pi,\tfrac12\pi$ (lower curves, with $\th=\tfrac1{20}\pi$ the lowest and $\th=\tfrac1{2}\pi$ the highest), and the curves $x=\pm 1/(2\sqrt{t})$ %curve $4tx^2=1$ 
(dashed lines).} %,\tfrac34\pi
\end{figure}}

An analogous situation arises for the solution $y_n(t;\th)$ given in \eqref{solyxn}, i.e.
\[y_n(t;\th)=\begin{cases} t^{1/2}+ \O\big(t^{-1}\big), &\quad{\rm if}\quad\th=0,\\
-t^{1/2}+ \O\big(t^{-1}\big), &\quad{\rm if}\quad\th\not=0\end{cases}\] 
(see Lemma \ref{lem43} below). We remark that Fornberg and Weideman \cite[Fig.~3 on p.~991]{refFW} plot the locations of the poles for the solution $w(z;\tfrac52)$ of \PII\ \eqref{eqPII}, which is equivalent to $y_2(t;\th)$, for various choices of $\th$. These plots show that the pole structure of the solutions is significantly different in the case when $\th=0$ compared to the case when $\th\not=0$. 
A study of ``Airy-type" solutions of \PII\ \eqref{eqPII} and $\mbox{P}_{\!34}$ \eqref{eqp34}, which are  equivalent to $y_n(t;\th)$ and $x_n(t;\th)$, for various choices of $\th$ is given in \cite{refClarksonAiry}.

\begin{lemma}\label{lem42x} If $x_1(t;\th)$ is given by \eqref{solx1}, then
\begin{equation} x_1(t;\th)=\begin{cases} \tfrac12\,t^{-1/2}+ \O\big(t^{-2}\big), &\quad{\rm if}\quad\th=0,\\
-\tfrac12\,t^{-1/2}+ \O\big(t^{-2}\big), &\quad{\rm if}\quad\th\not=0.\end{cases}\end{equation}
\end{lemma}
\begin{proof}The proof is very similar to that for Lemma \ref{lem42} so is left to the reader.
\end{proof}

\def\dotsn{\!\!\!\begin{array}{c}{\cdots}\\[-12pt]{}_{n}\end{array}\!\!\!}
\def\C{\mathcal{C}}
\section{\label{Sec5}Airy solutions of alternative discrete \p\ I}
Now we consider the case when $\th=0$, i.e. 
%\begin{subequations}
\begin{align} 
a_n(t)&=x_n(t;0)=-\deriv[2]{}{t}\ln {\Delta_n(t)},\qquad
b_n(t)=y_n(t;0)=\deriv{}{t}\ln \frac{\Delta_n(t)}{\Delta_{n+1}(t)},%\deriv{}{t}\ln \frac{\tau_n(t;0)}{\tau_{n+1}(t;0)},
\end{align}%\end{subequations}
where 
\begin{equation}\label{taun0}
\Delta_n(t)=\tau_n(t;0)%=\mathcal{W}\left(\Ai,\Ai',\ldots,\Ai^{(n-1)}\right)
=\det\left[\deriv[j+k]{}{t}\Ai(t)\right]_{j,k=0}^{n-1},%\qquad \tau_0(t;\th)=1,
%\qquad \Ai^{(m)}=\deriv[m]{}{t}\Ai(t),
\end{equation} with $\Delta_0(t)=1$. %$\tau_0(t;\th)=1$.
We remark that %the function $\tau_n(t;0)$ 
$\Delta_n(t)$ given by \eqref{taun0} arises in random matrix theory, in connection with the Gaussian Unitary Ensemble (GUE) in the soft-edge scaling limit, see e.g.\ \cite[p.~393]{ForrW01}.
Further, for $n\geq1$, %$\tau_n(t;0)$ 
$\Delta_n(t)$ has the multiple integral representation
\[%\tau_n(t;0)
\Delta_n(t)=\frac{(-1)^n}{(2\pi\i)^n}\int_{\C}%\dotsi\intS 
\dots\int_{\C} \prod_{j=1}^n \exp\left(\tfrac13x^3_j-tx_{j}\right)\prod_{1\leq k<\ell\leq n}(x_k-x_\ell)^2\,\d x_1\,\ldots\,\d x_n,\]
in which $\C$ is the standard Airy contour from $\infty\e^{-\pi\i/3}$ to $\infty\e^{\pi\i/3}$ \cite[p.~393]{ForrW01}.
%is the {partition function}.

\begin{lemma}\label{lem43}If $a_n(t)$ and $b_n(t)$ satisfies the recurrence relation \eqref{dPIsys} %\eqref{p2scrr} 
with
\begin{equation} a_0(t)=0,\qquad b_0(t)=-{\Ai'(t)}/{\Ai(t)},\end{equation}
where $\Ai(t)$ is the Airy function, then as $t\to\infty$
\begin{subequations}\begin{align} 
a_n(t)&=\frac{n}{2\,t^{1/2}}+ \O\big(t^{-2}\big),\qquad\mbox{for}\quad n\geq 1,\label{solan}\\
b_n(t)&=t^{1/2}+ \O\big(t^{-1}\big),\qquad\mbox{for}\quad n\geq 0.\label{solbn}
\end{align}\end{subequations}
\end{lemma}

\begin{proof}We shall first prove \eqref{solbn} by induction. If $b_n(t)=y_n(t;0)$ satisfies the recurrence relation \eqref{p2scrr}, which is a consequence of \eqref{dPIsys}, then $b_n(t)$ also satisfies the differential-difference system \eqref{p2scbt} and the differential equation \eqref{PIIscn}. We shall use the B\"acklund transformation \eqref{p2scbta}.

Clearly \eqref{solbn} is satisfied for $n=0$ from Lemma \ref{lem42}.
Now we assume as the inductive hypothesis $b_n(t)=t^{1/2}+ \O\big(t^{-1}\big)$, so we know from \eqref{solyna} that
\[b_{n}(t)=t^{1/2}+{\frac{2n+1}{4\,t}}- \frac{12n^2+12n+5}{32\,t^{5/2}}%\O\big(t^{-4}\big),\]
+\frac{(2n+1)(16n^2+16n+15)}{64\,t^{4}}+\O\big(t^{-11/2}\big),\]
and hence it is easily shown that
\begin{align*}
b_{n}^2(t)&=t+{\frac{2n+1}{2\,t^{1/2}}}-{\frac{2n^2+2n+1}{4t^2}}+ {\frac{5(2n+1)(4\,n^2+4n+5) }{64\,t^{7/2}}}
+\O\big(t^{-5}\big)\\
b_n'(t)&=\frac{1}{2\,t^{1/2}}-{\frac{2n+1}{4\,t^2}}+ {\frac{5(12\,n^2+12n+5) }{64\,t^{7/2}}}+\O\big(t^{-5}\big).
\end{align*}
Therefore
\begin{align*}b_n^2(t)+b_n'(t)-t
%&= {\frac{n+1}{t^{1/2}}}-{\frac{(n+1)^2}{2\,t^2}}+{\frac{5(n+1)(4n^2+8n+5)}{32\,t^{7/2}}}+\O\big(t^{-5}\big)\\
&={\frac{n+1}{t^{1/2}}}\left\{1 -{\frac{n+1}{2\,t^{3/2}}}+{\frac{5(4n^2+8n+5)}{32\,t^3}}+\O\big(t^{-9/2}\big) \right\},
\end{align*}
then
\begin{align*}
\frac{2(n+1)}{\displaystyle 
b_n^2(t)+b_n'(t)-t} %&=2\,t^{1/2}\left\{1 -{\frac{n+1}{2\,t^{3/2}}}+{\frac{5(4n^2+8n+5)}{32\,t^3}}+\O\big(t^{-9/2}\big) \right\}^{-1}\\
&=2\,t^{1/2}+\frac{n+1}{t}-\frac{12n^2+24n+17}{16\,t^{5/2}}+\O\big(t^{-4}\big),
\end{align*}
and so from \eqref{p2scbta}
\begin{align*} b_{n+1}(t) &= -b_n(t)+ \frac{2(n+1)}{b_n^2(t)+b_n'(t)-t}\\ &
%&=-\left\{t^{1/2}+{\frac{2n+1}{4\,t}}- \frac{12n^2+12n+5}{32\,t^{5/2}}+\O\big(t^{-4}\big)\right\} \\ &\qquad
%+2\,t^{1/2}+\frac{n+1}{t}-\frac{12n^2+24n+17}{16\,t^{5/2}}+\O\big(t^{-4}\big)\\
=t^{1/2}+\frac{2n+3}{4\,t}-{\frac{12n^2+36n+29}{32\,t^{5/2}} }+\O\big(t^{-4}\big),
\end{align*}
which is \eqref{solyna} with $n\to n+1$.
Hence the result \eqref{solbn} follows by induction. The result \eqref{solan} then follows immediately from \eqref{dPI-2}. \end{proof}

\begin{corollary}\label{cor52}
The positive solution $b_n(t)$ in Theorems \ref{thmab} and \ref{thmb} corresponds to the initial value $b_0(t) = -\Ai'(t)/\Ai(t)$.
\end{corollary}

{\begin{figure}[ht]
\[\begin{array}{cc}
\includegraphics[width=3in,height=2.5in]{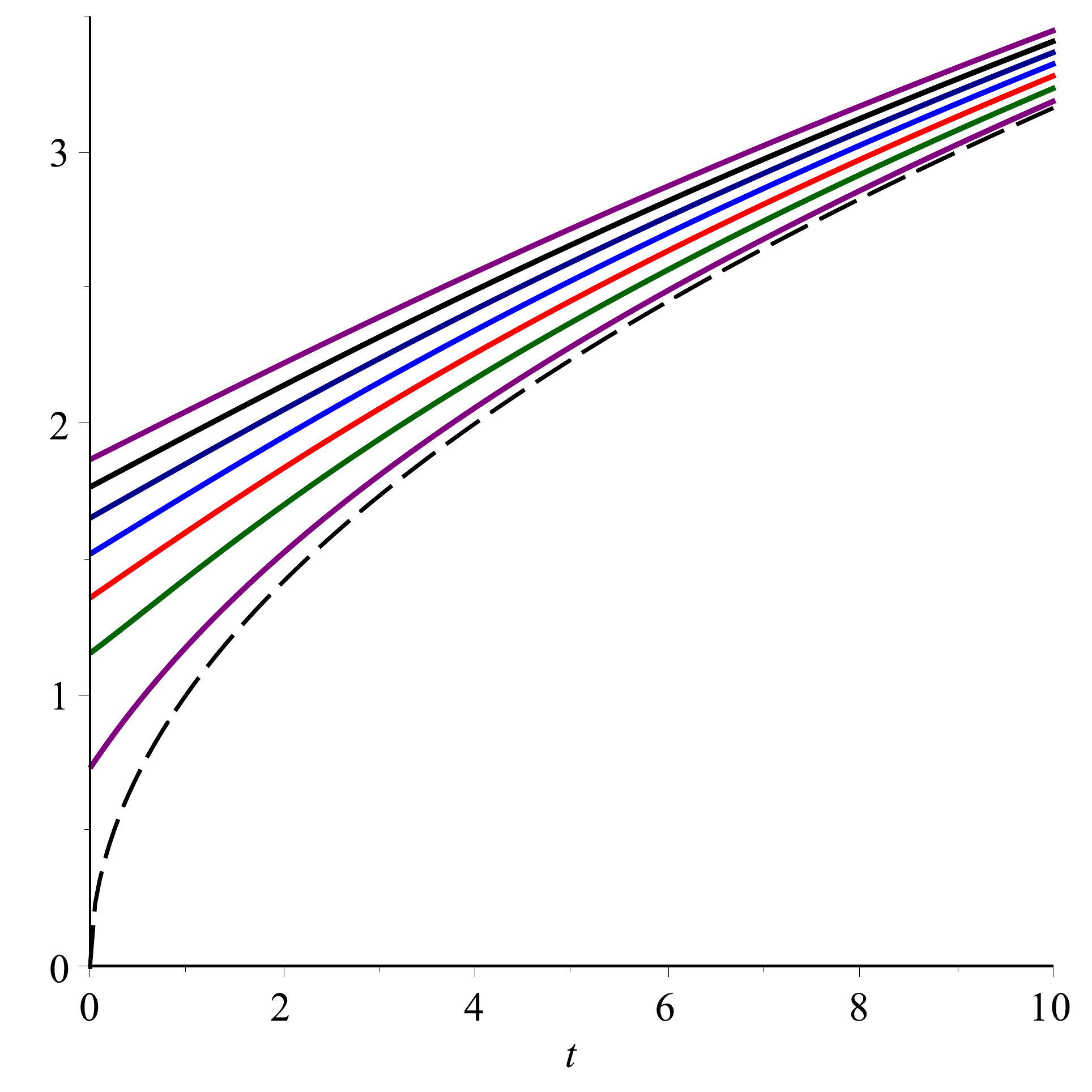} &\includegraphics[width=3in,height=2.5in]{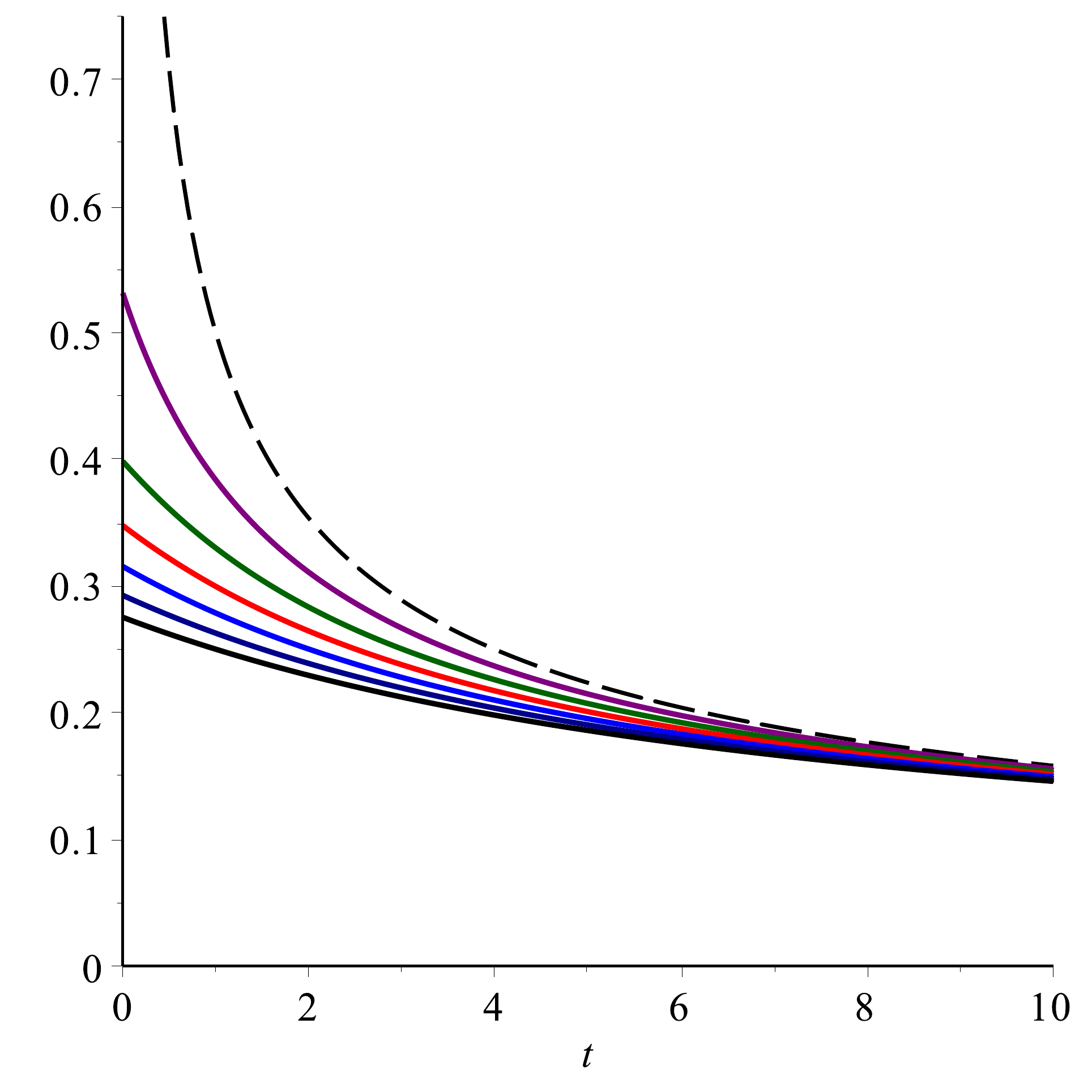}\\
%b_n(t)=y_n(t;0) & a_n(t)/n=x_n(t;0)
{\rm (i)} & {\rm (ii)} 
\end{array}\]
\caption{\label{plotxyn}(i) Plots of $b_n(t)=y_n(t;0)$ for $n=0,1,\ldots,6$, and the curve $y=\sqrt{t}$ (dashed line).
%} %,\tfrac34\pi
%\end{figure}}{\begin{figure}[ht]\[\includegraphics[width=3in]{altdP1x1x6} \]
%\caption{\label{plotxn}
(ii) Plots of  $a_n(t)/n=x_n(t;0)/n$ for $n=1,2,\ldots,6$, and the the curve $x= 1/(2\sqrt{t})$ %curve $4tx^2=1$ 
(dashed line).} %,\tfrac34\pi
\end{figure}}

Plots of $b_n(t)=y_n(t;0)$, for $n=0,1,\ldots,6$, %are given in Figure \ref{plotyn} 
and plots of  $a_n(t)/n=x_n(t;0)/n$, for $n=1,2,\ldots,6$, are given in Figure \ref{plotxyn}.
These plots suggest the following conjecture, which corroborate Conjecture \ref{conj32} and Lemma \ref{lem33}.

\begin{conjecture}
%\begin{enumerate}\item 
If $0<t_1<t_2$ then
\begin{equation} b_n(t_1)<b_n(t_2),\qquad a_n(t_1)>a_n(t_2),\end{equation}
i.e.\ $b_n(t)$ is monotonically increasing and $a_n(t)$ is monotonically decreasing for $t>0$. %\item 
For fixed $t$ with $t>0$ then
\begin{subequations}\begin{align} \sqrt{t} < b_n(t)&<b_{n+1}(t),\qquad \frac{1}{2\sqrt{t}}>\frac{a_n(t)}{n}>\frac{a_{n+1}(t)}{n+1}>0.\end{align}
\end{subequations}%\end{enumerate}
\end{conjecture}

%\section{\label{SecC}Conclusions}

\section*{Acknowledgements}
Walter Van Assche is supported by FWO research project G.0934.13, KU Leuven research grant OT/12/073 and the Belgian Interuniversity Attraction Poles program P7/18. He is grateful for the support he received to visit the University of Kent in September 2014 where this paper was initiated.
%\framebox{\parbox{5in}{Some extra work needs to be done here to show that this solution in terms of Airy functions (i.e.,with the choice of $C_2=0$) indeed gives the positive solution.}}

 \def\CUP{Cambridge University Press}
 
 \def\refjl#1#2#3#4#5#6#7{\vspace{-0.25cm}\bibitem{#1} \textrm{\frenchspacing#2}, \textit{#6}, 
{\frenchspacing#3}\ {#4} (#7), pp.~#5.}

\def\refjll#1#2#3#4#5#6#7{\vspace{-0.25cm}\bibitem{#1} \textrm{\frenchspacing#2}, \textit{#6}, 
{\frenchspacing#3}\ (#7), #5.}

\def\refbk#1#2#3#4#5{\vspace{-0.25cm}\bibitem{#1} \textrm{\frenchspacing#2}, \textit{#3}, #4, #5.} 

\def\refbkk#1#2#3#4#5{\vspace{-0.25cm}\bibitem{#1} \textrm{\frenchspacing#2}, \textit{#3}\ #4, #5.} 

\def\refcf#1#2#3#4#5#6{\vspace{-0.25cm}\bibitem{#1} \textrm{\frenchspacing#2}, \textit{#3},
in ``{#4}", {\frenchspacing#5}, #6.}

\end{document}